\documentclass[11pt]{article}

\usepackage{array}
\usepackage{amsfonts,amsmath,amssymb,amsthm,mathrsfs}
\usepackage{adjustbox}
\usepackage{tikz}
\usepackage{enumitem}
\usepackage{ifthen}
\usepackage{fullpage}

\usepackage[ruled]{algorithm2e} % For algorithms

\SetAlFnt{\small}
\SetAlCapFnt{\small}
\SetAlCapNameFnt{\small}
\SetAlCapHSkip{0pt}
\IncMargin{-\parindent}

\usepackage{booktabs} % For formal tables
\usepackage{xspace}
\usepackage[pagebackref=true,hidelinks]{hyperref}
\usepackage{bbm}
\usepackage{enumitem}
\usepackage[numbers]{natbib}
\usepackage{multirow}
\usepackage{graphicx}
\usepackage{subfigure}
\usepackage{wrapfig}
\usepackage{cleveref}

\newtheorem{theorem}{Theorem}[section]

\newtheorem{lemma}[theorem]{Lemma}
\newtheorem{corollary}[theorem]{Corollary}

\newenvironment{numberedtheorem}[1]{%
\renewcommand{\thetheorem}{#1}%
\begin{theorem}}{\end{theorem}\addtocounter{theorem}{-1}}

\newenvironment{numberedlemma}[1]{%
\renewcommand{\thetheorem}{#1}%
\begin{lemma}}{\end{lemma}\addtocounter{theorem}{-1}}

\newcommand{\rev}{\textsc{Rev}}
\newcommand{\srev}{\textsc{Srev}}

\newcommand{\E}{\mathbb{E}}

\newcommand{\mec}{\mathcal{M}}

\newcommand{\subadditiverev}{\textsc{SubAdditiveRev}}

\newcommand{\eps}{\epsilon}

\newcommand{\adaptalg}{\mathcal{A}}

\newcommand{\R}{\mathbb{R}}
\newcommand{\Z}{\mathbb{Z}}

\newcommand{\dist}{\mathcal{D}}

\newcommand{\remove}[1]{}

\newcommand{\orderedip}{\textsc{OrderedItemPricing}}
\newcommand{\feasiblesets}{\mathcal{S}}

\newcommand{\indsize}{\scriptsize}
\newcommand{\colind}[2]{\displaystyle\smash{\mathop{#1}^{\raisebox{.5\normalbaselineskip}{\indsize #2}}}}
\newcommand{\rowind}[1]{\mbox{\indsize #1}}

\newcommand{\range}{\mathcal{R}}
\newcommand{\vadd}{v^{\oplus}}
\newcommand{\vaddi}{v^{\oplus}_I}
\newcommand{\poly}{\operatorname{poly}}

\begin{document}

\title{Pricing Ordered Items}
\author{
Shuchi Chawla \\ University of Texas at Austin \\ {\tt shuchi@cs.utexas.edu} \and 
Rojin Rezvan \\ University of Texas at Austin \\ {\tt rojin@cs.utexas.edu} \and
Yifeng Teng \\ Google Research \\ {\tt yifengt@google.com} \and 
Christos Tzamos \\ University of Wisconsin-Madison \\ {\tt tzamos@wisc.edu}
}
\date{}

\maketitle
\thispagestyle{empty}

\begin{abstract}
%!TEX root = main.tex

We study the revenue guarantees and approximability of item pricing. Recent work \cite{chawla2019buy} shows that with $n$ heterogeneous items, item-pricing guarantees an $O(\log n)$ approximation to the optimal revenue achievable by any (buy-many) mechanism, even when buyers have arbitrarily combinatorial valuations. However, finding good item prices is challenging -- it is known \cite{chalermsook2013independent} that even under unit-demand valuations, it is NP-hard to find item prices that approximate the revenue of the optimal item pricing better than $O(\sqrt{n})$.

Our work provides a more fine-grained analysis of the revenue guarantees and computational complexity in terms of the number of item ``categories'' which may be significantly fewer than $n$. We assume the items are partitioned in $k$ categories so that items within a category are totally-ordered and a buyer's value for a bundle depends only on the best item contained from every category. 

We show that item-pricing guarantees an $O(\log k)$ approximation to the optimal (buy-many) revenue and provide a PTAS for computing the optimal item-pricing when $k$ is constant. We also provide a matching lower bound showing that the problem is (strongly) NP-hard even when $k=1$. Our results naturally extend to the case where items are only partially ordered, in which case the revenue guarantees and computational complexity depend on the width of the partial ordering, i.e. the largest set for which no two items are comparable.
\end{abstract}
\newpage

\setcounter{page}{1}

%!TEX root = main.tex

\section{Introduction}

% Item pricings are an important class of mechanisms. But their 
% approximability is not well understood. In this paper, we define 
% particular parameterization of value distributions that allows for 
% fine grained approximations to revenue optimal items pricings. 

% Describe settings. Examples. 

% Multi-item mechanism design is hard and recent work has focused on 
% fedex etc.. 

% Fedex exact result and other approx results. 

% Simple versus optimal. buy many literature. More details of fedex. Our 
% parameterized results. 

A dominant theme within algorithmic mechanism design is simplicity
versus optimality -- establishing that simple mechanisms can
approximate optimal ones within many settings. The simple mechanism in
most of these results is an {\em item pricing}, where the seller
determines a fixed price for each item and buyers can purchase any set
of items at the sum of the corresponding prices. Item pricings are
also an important class of mechanisms from a practical viewpoint --
most real world mechanisms are indeed item pricing mechanisms. 
However, despite their simplicity and popularity, 
finding good item prices for multi-item settings is a notoriously challenging problem
%approximating the revenue
%achievable with item pricing revenue over $n$ items is in general
and it is known to be
inapproximable within a factor better than $\sqrt{n}$ even for unit-demand buyers \cite{chalermsook2013independent}.

%As such
%the approximability and revenue guarantees of item pricings are
%interesting in their own right. However, the theoretical properties of
%item pricings are not yet well understood. (ADD SOMETHING ABOUT WHAT's
%KNOWN.) The 

In this paper, we focus on structured mechanism design instances and
perform a fine grained analysis of the approximability of item pricing
as well as its approximate optimality. We consider a standard
multi-parameter mechanism design setting where a revenue maximizing
seller offers multiple items for sale to a buyer whose value for the
items is drawn from a known distribution. We define a new
parameterization over value distributions wherein items can be
partitioned into a few categories and items within each category can
be ordered by desirability. We show that the number of categories
governs both the approximability and approximate optimality of item
pricings.

\paragraph{Ordered items and the approximability of item pricing.}
At the heart of our parameterization is the so-called FedEx Problem
that was first studied by Fiat et al~\cite{fiat2016fedex}. In the FedEx Problem, the
items offered by the seller correspond to shipping times for a
package; each buyer has a deadline for shipping their package and
obtains a fixed value if the shipping time meets their deadline. The
FedEx Problem occupies a sweet-spot between single-parameter mechanism
design settings where a buyer's preferences can be fully described
through a scalar value; and multi-parameter settings where different
(sets of) items bring the buyer different values. Accordingly it
exhibits some but not all of the complexity of multi-parameter
settings. Indeed, as we show, in contrast to the general case, the optimal item pricing for Fedex instances can be computed in polynomial time.

The FedEx Problem is a special case of ``totally ordered'' settings
where items can be ranked by quality and {\em every} buyer type weakly
prefers a higher ranked item to a lower ranked one.\footnote{In the
  FedEx setting, for example, every buyer weakly prefers earlier
  shipping times to later ones.} More generally, we
consider settings where items can be partitioned into $k$ categories
such that within each category items are totally ordered by
quality:

\begin{quote}
 Consider, for example, a car dealership that sells $k$
different models of cars. Each model comes in a variety of different
trims -- the most basic trim along with a sequence of upgrades. For
any particular model or category, every buyer has the same ordering of
values over different trims although values differ arbitrarily across
buyers and across categories. One buyer may value the
luxury trim \$5,000 higher than the standard trim and another may
value them the same, but no buyer values the standard trim more than
the luxury trim. 

For another example, consider an internet service
provider such as Comcast, AT\&T or Spectrum that offers multiple
products such as TV, internet, and phone service. Each individual
product has quality or service levels that are ordered. In particular,
every buyer weakly prefers higher internet speeds to lower speeds and
unlimited talk time to limited talk time. However, buyers may assign
different values to different combinations of the three services. 
\end{quote}

We
emphasize that both the totally ordered setting and the $k$-category
setting are multi-parameter settings where the buyer's values are
combinatorial and described as functions over the set of items
allocated to the buyer. Beyond the ordering over items within each
category, we make no assumptions on the buyer's values over sets of
items.

Our main computational finding is that the approximability of item
pricing is governed by the parameter $k$.
%For the FedEx setting, there is a polynomial time
%algorithm that computes the revenue-optimal item pricing exactly. 
For
the totally ordered ($k=1$) and $k$-category settings, we provide a
polynomial time approximation scheme with a running time that depends
exponentially on $k$. For any given $\eps>0$, our algorithm
returns an item pricing
that approximates the revenue of the optimal item pricing within a
factor of $(1+\eps)$ and runs in time $\poly(m, n^{\poly(k/\eps)}, b)$ where $n$ is the number of items, $m$ is the support of the distribution and $b$ is
the bit complexity of buyer's value distribution. Our approximation
scheme is the best possible, as we show that finding the optimal item
pricing is strongly NP-hard even for $k=1$. Our algorithm is
particularly relevant and useful when $k$ is a small constant such as in the
examples described above.

\begin{theorem}\label{thm:intro-ptas}
For any distribution of support size $m$ over $k$-category valuations $v: 2^{[n]} \rightarrow [1,2^b]$, we can compute an $(1+\eps)$-approximate item pricing  in $\text{poly}(m,n^{poly(k,1/\eps)},b)$ time.
\end{theorem}

\paragraph{The approximate optimality of item pricing.} As aforementioned,
a central problem in multi-item mechanism design is approximating the revenue
of the optimal mechanism in multi-parameter settings by simple mechanisms like item pricing.
%Approximating
%the revenue of the optimal mechanism in multi-parameter settings is a
%notoriously challenging problem. 
In fact, in the kinds of settings we
study in this paper (with no assumptions on the value distributions),
it is known that no simple mechanisms can provide any finite approximation
to the optimal revenue in the worst case. This general case
impossibility of simple-versus-optimal results has led to two
complementary lines of work in recent years. 

The first looks at structured settings for which approximately-optimal
mechanisms can be characterized. The FedEx
problem~\cite{fiat2016fedex, saxena2018menu} and its extensions to
so-called ``interdimensional'' settings \cite{devanur2017optimal, devanur2020optimal}
belong to this line of work; in these settings, the optimal mechanism
can have an exponential or even unbounded description complexity but
under appropriate assumptions, mechanisms with polynomial menu size
provide an approximation. Another series of
works~\cite{babaioff2014simple, chawla2007algorithmic,
  rubinstein2018simple} bounds the revenue gap between item pricings
and optimal mechanisms assuming that the buyer' values are subadditive
and independent across different items.

The second line of work places an extra incentive constraint on the
revenue maximization problem. Instead of viewing a mechanism
as a one time interaction between the seller and a buyer, it is assumed that
the buyer can visit the mechanism multiple times purchasing different bundles
of items. In this ``buy-many" setting, complicated mechanisms that
extracted arbitrarily higher revenue than simpler ones are no longer incentive compatible
as the buyer can buy multiple cheaper options instead of a single expensive one.
In fact, recent work \cite{chawla2019buy} shows that item-pricing achieves a $\Theta(\log n)$ approximation
to the optimal buy-many mechanism and this is tight in a strong sense
as no simple mechanism, i.e. one with polynomial description complexity, 
can approximate the optimal revenue better than a logarithmic factor.

Our work unifies the two approaches and considers the revenue approximation of item pricing in more structured buy-many settings.
Our first finding is that item pricing is the optimal buy-many mechanism in the FedEx setting.
More generally, we find that the revenue guarantees of item pricing are again governed by the parameter $k$ of our parameterization.
In the totally ordered setting where $k=1$, we show that item pricing is no longer optimal but achieves a constant factor approximation to the
optimal buy-many revenue. For $k$ categories, we show that the approximation
is $\Theta(\log k)$. This gives a smooth degradation of the revenue guarantee as the instances become less and less structured.

\begin{theorem}\label{thm:intro-buymany}
For any distribution over $k$-category valuation functions, the optimal item pricing guarantees a $1/\Theta(\log k)$ fraction of the revenue achievable by the optimal buy-many mechanism.
\end{theorem}

\paragraph{Implications for Buy-Many Mechanism Design.}
Even though our focus in this work is on item pricing and its revenue guarantees, our result gives the first computationally efficient
algorithm for computing approximately optimal buy-many mechanisms in structured settings. In
contrast to the setting of buy-one mechanisms 
where the optimal mechanism can typically be computed via a linear program of polynomial size in the support of the distribution of values, 
no such algorithm is known for buy-many settings. In fact, we observe that the $\sqrt{n}$ inapproximability of item-pricing even for unit-demand settings, directly
implies a $\sqrt{n} / \log(n)$ inapproximability for buy-many mechanisms as one can efficiently convert any buy-many mechanism into an item pricing one
with a logarithmic loss in approximation. Our results show that for structured settings, the optimal buy-many mechanism is efficiently approximable and that 
such an approximation can be achieved via item pricing.

We remark that being able to obtain approximate item pricing or buy-many mechanisms is important even in cases where the optimal buy-one mechanism might be easier to compute.
This is because buy-one mechanisms may be inherently complex and
difficult for the buyers to understand and participate in. 
More significantly, in many settings, it may be unrealistic to expect that the revenue promised by a buy-one mechanism is achievable in practice.
For cases like shopping from a retail store, it may not be feasible to implement a buy-one mechanism as buyers faced with superadditive prices would break their desired bundle into smaller ones visiting the store multiple times. This would result in significantly lower revenue than expected by the buy-one model.

\paragraph{Extensions.}
We further consider settings where there is a partial ordering over
items. Consider, for example, an electronics company that
manufactures both cameras and cell phones. Some cell phones capture
all of the features of certain cameras, and therefore all buyers
weakly prefer the former to the latter. But not all cameras and cell
phones are comparable. 
We say that an item $i$ dominates another item 
$j$ if for every set $S$ of items containing both $i$ and $j$, every buyer is indifferent between getting $S$ or $S \setminus \{j\}$.
%\ytnote{Should we define it the other way
%around, that if $i$ dominates $j$, then every buyer weakly prefers
%$S\cup\{i\}\setminus\{j\}$ to the set $S$?} 

We use the parameter $k$ to
denote the ``width'' of the partial ordering over items---the size of
the largest set of incomparable items or the longest anti-chain in the
partial ordering. Note, that the $k$-category setting is a special case of this more general width-$k$ setting.
Our PTAS for item pricing of Theorem~\ref{thm:intro-ptas} as well as the buy-many revenue approximation result
 of Theorem~\ref{thm:intro-buymany} naturally extend to this more general setting with the same guarantees.

A more relaxed condition for partial ordering across items specifies
that item $i$ dominates another item $j$ if all sets of items $S$ that
do not contain items $i$ or $j$, adding $i$ to $S$ is always
preferable to adding $j$. Unfortunately, we show that under such a
weak condition, pricing cannot guarantee a constant fraction of the
optimal buy-many revenue even in simple settings. In fact, even with additive buyers over totally ordered items,
we show that no buy-many
mechanism with polynomial description complexity can achieve better than $1/o(\log \log n)$ fraction of the optimal buy-many revenue (see Section~\ref{sec:totally-ordered-additive}). It is an interesting open question left by our work to show that this bound is indeed achievable by item pricing.

 \paragraph{Our techniques.}
 Our techniques are easiest to understand in the context of a
 unit-demand buyer with totally ordered items. Our analysis of the gap
 between item pricings and optimal buy-many mechanisms in this setting
 hinges on a characterization of the buyer's optimal buy-many
 strategy. Faced with a menu of randomized options, the buyer
 essentially behaves like a Pandora's box algorithm which at every
 step opens a box (i.e. purchases a lottery) and obtains a random
 reward. Because the same lotteries can be purchased any number of
 times, the buyer's optimal strategy is to pick a single lottery
 repeatedly until an item of a certain minimum value is
 instantiated. This characterization allows us to relate the buyer's
 utility to the value of the item(s) bought by the buyer. We can then
 apply a lemma from \cite{chawla2019buy} that relates the revenue
 obtained by an item pricing to the change in the buyer's utility at
 different scalings of that item pricing.

In order to approximate the optimal item pricing for a unit demand
buyer with totally ordered items, we view the buyer as additive over
item {\em upgrades}: the purchase of an item $i$ can be viewed
equivalently as the purchase of the base item $1$ along with a series
of upgrades, $1$ to $2$, $2$ to $3$, and so on till $i$. The benefit
in doing so is that with some slight loss in approximation, we can
group upgrades into different pricing scales, and price each scale
independently.  This permits a dynamic programming based algorithm
for optimizing the prices of the upgrades. The pricing found in this
manner can be easily converted into an item pricing with the same
revenue.

\subsection{Other related work}

% - Follow ups to Fedex including \cite{saxena2018menu} Devanur et al. EC'20 \cite{devanur2020optimal}

% - Unlimited supply settings approx to item pricing, edge pricing etc. Chalermsook et. al FOCS'13 \cite{chalermsook2013independent}, Lee STOC'15 \cite{lee2015hardness}

% - Hardness of item pricing \cite{chen2014complexity}

% - PTAS to optimal mechanisms for unit demand buyers; Cai Daskalakis
% 2015 \cite{cai2011extreme}, Kothari et al. FOCS'19 \cite{kothari2019approximation}

% - Buy-many

The computational complexity of item pricing for a single buyer has
been studied previously for a variety of valuation functions. One
widely studied setting is the $k$-hypergraph pricing problem, where
each possible realization of the buyer is unit-demand over a set of at
most $k$ items. It has been shown that there exists an algorithm with
competitive ratio $O(\min(k,\sqrt{n\log n}))$
\cite{chalermsook2013independent} (also see
\cite{guruswami2005profit,briest2006single,balcan2006approximation}),
and is hard to approximate within
$\Omega(\min(k^{1-\eps},n^{1/2-\eps}))$ under the Exponential Time
Hypothesis \cite{chalermsook2013independent} (also see
\cite{briest2008uniform,chalermsook2012improved,chalermsook2013graph}). Such
results also extend to a single-minded buyer that wants an entire set
of at most $k$ items. The specific case where $k=2$ is called the
graph vertex pricing, for which there is an efficient algorithm with competitive ratio 4 \cite{balcan2006approximation}. No efficient algorithm can give an approximation ratio better than 4 assuming the Unique Games Conjecture \cite{lee2015hardness} (also see \cite{guruswami2005profit,khandekar2009hardness}). Another special case is the tollbooth problem, where the buyer demands a path on a path graph. This problem is strongly NP-hard \cite{elbassioni2009profit}, and a PTAS is known \cite{grandoni2016pricing} (also see \cite{balcan2006approximation,gamzu2010sublogarithmic}).

Another line of work studies the problem of selling to a unit-demand
buyer with item values drawn from independent distributions. For
general distributions, computing the optimal item pricing is NP-hard
\cite{chen2014complexity}. The optimal item pricing revenue can be
approximated to within a factor of 2 (providing a 4-approximation to
the optimal revenue overall) \cite{chawla2007algorithmic,chawla2010multi}, and a PTAS (or QPTAS) exists if the item values are drawn from monotone hazard rate (or regular) distributions \cite{cai2011extreme}. The problem of finding the revenue from the optimal mechanism for a unit-demand buyer with independent item values has been further studied: it is known that no efficient exact algorithm exists unless the polynomial-time hierarchy collapses \cite{chen2015complexity}, and a QPTAS exists \cite{kothari2019approximation}.

The recent decade has seen much work on approximating the optimal
revenue via simple mechanisms such as item pricing and grand bundle
pricing: for a single unit-demand buyer
\cite{chawla2007algorithmic,chawla2015power}; an additive buyer
\cite{hart2013menu,li2013revenue,babaioff2014simple}; a subadditive
buyer \cite{rubinstein2018simple,chawla2016mechanism}; as well as for
multi-buyer settings
\cite{chawla2010multi,yao2014n,chawla2016mechanism,cai2017simple,duettinglog}. All
of these results require independence across individual item
values. For correlated item values, simple mechanisms cannot provide
finite approximations to the optimal revenue and bounded gaps are only
known in comparison to the optimal {\em buy-many} revenue.
\cite{briest2015pricing} shows that item pricing gives $O(\log
n)$-approximation to the optimal buy-many revenue for a unit-demand
buyer. \cite{chawla2019buy} further generalizes the result to a
general-valued buyer, and \cite{chawla2020menu} characterizes the
tight menu-size complexity of the mechanism needed for $(1+\eps)$-approximation in revenue.

% The study of optimal mechanism design in the ordered-item setting has attracted attention recently. For totally-ordered items, there has been study on the FedEx problem \cite{fiat2016fedex,saxena2018menu} and the multi-unit pricing problem \cite{devanur2020optimalmultiunit}. For partially-ordered items, \cite{devanur2020optimal} studies the menu-size complexity for the optimal mechanism for a single-minded buyer. \ytnote{This may appear earlier in the intro, and get discussed in more detail, as suggested in the previous paragraphs?}

% \ytnote{Maybe shouldn't cite other 1.5-dimensional mechanism design problem like \cite{devanur2017optimal}? Does not seem relevant to me.}

\subsection{Organization}

We present our results by iteratively building from the simplest case
of the FedEx-Problem in Section \ref{sec:fedex}, to the case of
totally-ordered items in Section \ref{sec:totally-ordered} and
finally the general case with partially ordered items in Section
\ref{sec:partially-ordered}. Some proofs are deferred to the appendix.

\section{Definitions}
\label{sec:prelim}

We study the multidimensional mechanism design problem where the seller has $n$ heterogeneous items to sell to a single buyer, and aims to maximize the revenue. The buyer's value type is specified by a valuation function $v:2^{[n]}\to\R_{\geq 0}$ that assigns a non-negative value to every set of items. The valuation functions are monotone: for any $S,T\subseteq [n]$ with $S\subseteq T$ and any valuation function $v$, we have $v(S)\leq v(T)$. We study the \textit{Bayesian setting}, where the buyer's valuation function $v$ is drawn from a publicly known distribution $\dist$ over the set of all monotone valuation functions.

\paragraph{Unit-demand Buyers.} We say that a buyer is \textit{unit-demand} over all items, if the buyer is only interested in purchasing one item, and her value for any set of items is solely determined by the item that is most valuable to her. In other words, for any set $S\subseteq [n]$, $v(S)=\max_{i\in S}v(\{i\})$. When there is no ambiguity, we use $v_i$ to denote $v(\{i\})$ for a unit-demand buyer of type $v$.

% \paragraph{Totally-ordered Items.} Say a unit-demand buyer has \textit{totally-ordered} value for all items, if for each value realization $v$ of the buyer, she has an identical ordering for the value of all items. In other words, there exists a permutation $(i_1,\cdots,i_n)$ of $(1,2,\cdots,n)$, such that for any buyer type $v$, $v_{i_1}\leq v_{i_2}\leq\cdots v_{i_n}$. Without loss of generality, we can simply assume $v_{1}\leq v_{2}\leq\cdots\leq v_{n}$ for a buyer with totally-ordered item values.

\paragraph{Totally-ordered Items.} We say that a unit-demand buyer has \textit{totally-ordered} values, if for every possible value realization $v$ of the buyer, $v_{1}\leq v_{2}\leq\cdots\leq v_{n}$.

\paragraph{Partially-ordered Items.} 
Let $\preceq$ denote a partial ordering over the $n$ items. We say that the buyer has \textit{partially-ordered} values with respect to the relation $\preceq$ if for every realizable valuation function $v$, every pair of items $i$ and $j$ with $i\preceq j$, and every set $S\subseteq [n]$, we have $v(S\cup\{i,j\}) = v(S\cup\{j\})$. We say that the item $j$ dominates $i$.  In other words, the buyer may discard from his allocated set any item that is dominated by another item in his allocation with no loss in value. As a consequence, the only ``interesting'' allocations over partially-ordered items are sets that form {\em antichains}, i.e. where no two items are comparable. An important parameter of a partially ordered set is its \textit{width} that is defined to be the size of the largest antichain. We use $k$ to denote the width of the partial ordering $\preceq$. 
An important special case of partially-ordered items is the $k$-category setting where items are partitioned in $k$-categories. In this setting, items within a category are totally ordered and every buyer's value for a bundle depends only on the best item of each category it contains.

\paragraph{Input Model for the Computational Problem.} When we study computational problems, we assume that the input distribution $\dist$ is provided explicitly over a support of size $m$. Each buyer type $v$ in the support is a vector of size $O(n^k)$ that specifies the buyer's value $v(T)$ for all possible sets $T$ of size at most $k$,\footnote{Note that it suffices to specify the buyer's value over sets of size at most $k$, where $k$ is the width of the partial ordering over items, because the buyer only desires sets that form antichains.} and is accompanied with a probability of realization $\Pr[v]$. We further assume that the value $v(T)$ of each set of items is either $0$, or in range $[1,\range]$. Without loss of generality we assume that each buyer type $v$ in the support $\dist$ is non-trivial: $v([n])\geq 1$.

\paragraph{Single-buyer Mechanisms.} By the Taxation Principle \cite{rochet1985taxation}, any single-buyer mechanism can be described as a menu of possible outcomes, and the buyer can select one menu option. Each outcome $\lambda=(x,p)\in \Delta(2^{[n]})\times \R_{\geq 0}$ is a lottery that is specified by a randomized allocation $x$ over the sets of items, and a price $p$ that is the payment of the buyer if she wants to get such an allocation. For any set $S\subseteq[n]$, $x_S$ denotes the probability that only items in set $S$ are allocated to the buyer, and we have $\|x\|_1=1$. We will use $x(\lambda)$ and $p(\lambda)$ to denote the allocation and the payment of any lottery $\lambda$. For any buyer of valuation function $v$, her \textit{value} for lottery $\lambda$ is defined by $v(\lambda)\equiv \E_{S\sim x}v(S)$; her \textit{utility} for purchasing $\lambda$ is defined by $u_v(\lambda)\equiv v(\lambda)-p(\lambda)$. We will also use $S\sim \lambda$ to denote a set of items drawn from set distribution $x(\lambda)$. 

Given a mechanism $\mec$ with a menu of lotteries $\Lambda$, the buyer selects the menu option $\lambda$ that maximizes her utility $u_v(\lambda)$. When there are multiple lotteries with the same highest utility for the buyer, the seller can choose the most expensive lottery to sell to the buyer. Without loss of generality, we assume that for any allocation $x\in \Delta(2^{[n]})$ over the sets of items, there is a corresponding price $p(x)$ such that $(x,p(x))\in\Lambda$. We also use the pricing function $p$ as an alternative definition of the mechanism $\mec$. The buyer's utility is defined as $u_p(v)=v(x)-p(x)$. The buyer's payment is $\rev_p(v)=p(x)$, and we write the revenue of mechanism $p$ as $\rev_p=\E_{v\sim \dist}\rev_p(v)$. Since the mechanism only allows the buyer to interact with the mechanism for once, it is also called \textit{buy-one mechanism}. 

\paragraph{Buy-many Mechanisms.} In an (adaptively) buy-many mechanism, the buyer is allowed to interact with the mechanism for multiple times. To be more precise, a buy-many mechanism $\mec$ generated by a set $\Lambda$ of lotteries can be defined as follows. The buyer can adaptively purchase a (random) sequence of lotteries in $\Lambda$, which means that in each step, the buyer can decide which lottery to purchase given the instantiation of the previous lotteries in the sequence. The buyer gets the union of all items allocated in each step and pays the sum of the prices of all purchased lotteries. For any adaptive algorithm $\adaptalg$, define $\Lambda_{\adaptalg}=(\lambda_{\adaptalg,1},\lambda_{\adaptalg,2},\cdots)$ to be the random sequence of lotteries purchased by the buyer of type $v$. The expected value of the buyer is 
\[v(\Lambda_{\adaptalg})\equiv\E_{(S_1,S_2,\cdots)\sim (\lambda_{\adaptalg,1},\lambda_{\adaptalg,2},\cdots)}v\left(\bigcup_{i\geq 1} S_i\right),\]
and the payment of the buyer is 
\[p(\Lambda_{\adaptalg})\equiv\E_{\adaptalg}\sum_{i\geq 1}p(\lambda_{\adaptalg,i}).\]
Any buy-many mechanism can be described by a \textit{buy-one} menu, where the buyer is only allowed to purchase a single lottery. This is because the expected outcome of any adaptive algorithm $\adaptalg$ can be described by the allocation $\cup_i(S_i\sim \lambda_{\adaptalg,i})$, and an expected payment $p(\Lambda_{\adaptalg})$. We say that a buy-one menu $\Lambda$ satisfies the \textit{buy-many constraint}, if for every adaptive algorithm $\adaptalg$, there exists a cheaper single lottery $\lambda\in\Lambda$ dominating it. Rigorously speaking, there exists $\lambda\in \Lambda$ with $p(\lambda)\leq p(\Lambda_{\adaptalg})$ such that there exists a coupling between a random draw $S$ from $\lambda$, and the union of random draws $S'$ from $\Lambda_{\adaptalg}$, satisfying $S\supseteq S'$. Intuitively, a buy-one menu satisfies the buy-many constraint, if the buyer always prefers to purchase a single option from the menu, even if she has the option to adaptively interact with the mechanism for multiple times. In later sections, when we refer to a ``buy-many mechanism'' with menu $\Lambda$, we are always referring to a buy-one mechanism with menu $\Lambda$ that satisfies the buy-many constraint.

\section{Warm-Up: Item Pricing in the FedEx Setting}
\label{sec:fedex}

In this section, we study the item pricing in the FedEx setting \cite{fiat2016fedex}. The buyer's value distribution in the FedEx problem has the following structure. Any buyer type $v$ is defined by the pair of parameters $(i_v,v_H)$ with $v_i=0$ for $i<i_v$ and $v_i=v_H$ otherwise. In other words, the buyer is totally-ordered and has at most two distinct values for all items, with the lower value being 0.

\subsection{The optimality of item pricing}

Our first observation is that item pricing achieves the optimal revenue obtained by any buy-many mechanism. 

\begin{theorem}\label{thm:fedex-optimal}
For any value distribution in the FedEx setting, there exists an item pricing that achieves the optimal buy-many revenue.
\end{theorem}

\begin{proof}
  Consider a buyer with value function $v$ in the FedEx setting. Recall that the buyer only values items with index $\ge i_v$ and values all of them equally. Therefore the buyer obtains the same value from an allocation $x=(x_1,x_2,\cdots,x_n)$ as from an allocation $x'$ where $x'_{i_v}=\sum_{i\ge i_v} x_i$ and $x'_i=0$ for $i\ne i_v$.

Given any buy-many menu $\{(x,p)\}$, consider replacing every lottery $(x,p)$ with $n$ different options: $(x^{(1)},p^{(1)})=((\sum_{j\geq 1}x_j,0,\cdots,0),p)$, $(x^{(2)},p^{(2)})=((0,\sum_{j\geq 2}x_j,0,\cdots,0),p)$, $\cdots$, $(x^{(n)},p^{(n)})=((0,\cdots,0,x_n),p)$. By our observation above, for every buyer type $v$, one of the $n$ new options bring the same utility to the buyer as $(x,p)$ and all other options bring lower utility. As a result, the new mechanism is identical in its allocations and revenue to the original one.

  % Notice that any buyer $v$ has value 0 for items with index smaller than $i_v$, and identical values for items with index at least $i_v$. Thus for any allocation $x=(x_1,x_2,\cdots,x_n)$, the buyer $v$ views it as $x'=(0,0,\cdots,0,x'_{i_v},0,\cdots,0)$, where $x'_{i_v}=\sum_{j\geq i_v}x_{i_j}$. Therefore, in the optimal buy-many mechanism, we can replace any lottery $(x,p)$ with $n$ lotteries $(x^{(1)},p^{(1)})=((\sum_{j\geq 1}x_j,0,\cdots,0),p)$, $(x^{(2)},p^{(2)})=((0,\sum_{j\geq 2}x_j,0,\cdots,0),p)$, $\cdots$, $(x^{(n)},p^{(n)})=((0,\cdots,0,x_n),p)$, without changing the outcome of the mechanism.

Observe that the new mechanism sells each item separately (but with different probabilities of allocation). For such a mechanism, we claim that every buyer type $v$ purchases a lottery that sells item $i_v$ with allocation 1: if the buyer purchases a lottery $\lambda$ that sells item $i_v$ with probability $x_{i_v}$, the buyer can repeatedly purchase the same lottery until she gets the item, which increases her utility. We may therefore drop any options that allocate items with probability less than $1$ from the menu without changing the allocations or revenue of the mechanism. This final mechanism is an item pricing, and the theorem follows.
%Thus a buy-many mechanism that sells each item separately is equivalent to an item pricing. This finishes the proof that item pricing can achieve the optimal revenue obtained by any buy-many mechanism.
\end{proof}

\subsection{A poly-time algorithm for finding optimal item pricings}

In this section, we show that the optimal item pricing in the FedEx setting can be computed efficiently. We actually prove a stronger result: for each realized buyer type $v$, if the buyer has at most two distinct item values, $\orderedip$ can be solved in polynomial time via dynamic programming. For each buyer type $v$, let $v_L$ and $v_H$ denote the two different item values in $v$, and let $i_v$ be the smallest item type with item value $v_{H}$. In other words, $v_1=v_2=\cdots=v_{i_{v}-1}=v_L$, and $v_{i_v}=v_{i_v+1}=\cdots=v_n=i_H$. If $v_1=v_n$, we define $i_v=1$ and $v_L=v_H=v_1$. The FedEx Problem is a special case with $v_L=0$.

\begin{theorem}\label{thm:fedex-computation}
In the totally-ordered setting, if each realized buyer type has at most two distinct item values, then the optimal item pricing can be computed in polynomial time.
\end{theorem}

\begin{proof}
Let $\Pr[v]$ be the realization probability of $v$ under input value distribution $\dist$. Without loss of generality, we only study item pricings with monotone item prices $p_1\leq p_2\leq\cdots\leq p_n$.
For a buyer type with $v_H>v_L$, the buyer would either purchase item 1, or item $i_v$, or nothing. For a buyer type with $v_H=v_L$, the buyer would either purchase item 1, or nothing. 

To compute the optimal item pricing, we first find a set of feasible prices for each item, then use a dynamic program to find the optimal item pricing. Define $\Pi_L=\{v_1|v\sim \dist,i\in[n]\}\cup\{0\}$ be the set of all possible values for item 1, including 0. Let $\Pi^*=\{z|z=v_H-v_L+y,y\in \Pi_L,v\sim \dist\}\cup \Pi_L$. We first observe that we may restrict prices to lie in a set of polynomial size without loss in revenue. The proof of this lemma is deferred to Section~\ref{sec:fedex-appendix}.

\begin{lemma}\label{lem:two-item-value-lemma}
There exists an optimal item pricing such that $p_1\in \Pi_L$, and $p_i\in \Pi^*$ for each $i\geq 2$.
\end{lemma}

Now we are ready to find the optimal item pricing. Let $F[y,i,z]$ denote the total revenue from buyer types $v$ with $i_v\leq i$, under a monotone item pricing that has already priced the first $i$ items, with $p_1=y$ and $p_i=z$. Then we have the following recursive formula:
\begin{equation*}
F[y,i,z]=\max_{z'\leq z, z'\in \Pi^*}\Big\{F[y,i-1,z']+\sum_{v:i_v=i}\Pr[v]G[v,y,z]\Big\},
\end{equation*}
where $G[v,y,z]$ is the payment of buyer type $v$ with item price $y$ for item 1, and price $z$ for item $i_v$. In other words, $G[v,y,z]=z$ if $v_H-z\geq v_1-y$ and $v_H\geq z$; $G[v,y,z]=y$ if $v_1-y>v_H-z$ and $v_1\geq y$; $G[v,y,z]=0$ otherwise. The recursive formula is based on the following fact: if $p_1$ is fixed, the revenue contribution of buyer types with $i_v=i$ only depends on $p_i$. The optimal item pricing revenue we want to compute is $\max_{y\in \Pi_L,z\in \Pi^*,z\geq y}F[y,n,z]$. Since the table has a polynomial number of entries, and the inductive steps can be computed in polynomial time, the total running time is also polynomial in the number of items and the support size of the distribution.
\end{proof}

\section{Item Pricing in the Totally-Ordered Setting}
\label{sec:totally-ordered}

In the previous section we observed that for the FedEx setting, item pricing is not only optimal but also polynomial-time computable. When considering the general totally-ordered setting, both properties no longer hold. We show that for a general-valued buyer, item pricings may achieve strictly less revenue than the optimal buy-many mechanism. We complement this result by showing that item pricing gives a constant approximation in revenue to the optimal buy-many mechanism. Next, we show that computing the optimal item pricing in the totally-ordered setting is strongly NP-hard, thus there is no FPTAS algorithm finding the revenue obtained by the optimal item pricing. We complement the hardness result by providing a PTAS computing an approximately optimal item pricing, thus giving a tight characterization of the computational complexity of the problem.

\subsection{Item pricing is a constant approximation to the optimal buy-many revenue}
\label{sec:totally-ordered-approx}

We first provide an example which shows that in the totally-ordered setting, the optimal item pricing and the optimal buy-many mechanism may have a constant factor revenue gap. Consider the following example: Let there be $2$ items and $3$ unit-demand buyers, with the following values for items $1$ and $2$ respectively, each realized with probability $\frac{1}{3}$:
\begin{equation*}
	v^{(1)}=(0,5), v^{(2)}=(1,3), v^{(3)}=(1,2).
\end{equation*}
The optimal buy-one mechanism has the following menu:
\begin{equation*}
	\lambda_1=((0,1), 5), \lambda_2=\left(\left(\frac{2}{3}, \frac{1}{3}\right), \frac{5}{3}\right), \lambda_3=((1,0),1).
\end{equation*}

The lotteries are written in the form of $((x_1,x_2),p)$ where $x_1$ and $x_2$ are the probabilities of the buyer getting items $1$ or $2$ respectively, and $p$ is the price for this lottery. In the mechanism, each buyer $v^{(i)}$ prefers to purchase lottery $\lambda_i$. Observe that the menu of lotteries satisfies the buy-many constraint. This is because to achieve the allocation of any $\lambda_i$ using the other two lotteries in the menu, one always needs to pay more than $p_i$. Thus the mechanism is also the optimal buy-many mechanism, with revenue $\frac{23}{9}$. Observe that the optimal item pricing for this instance is $p_1=1,\ p_2=3$, which yields a revenue of $\frac{7}{3}<\frac{23}{9}$. Thus there can be a constant gap between the optimal item pricing revenue and the optimal revenue obtained by any (buy-many) mechanism.

Then we show that item pricings can actually achieve a constant fraction of the revenue obtained by the optimal buy-many mechanism.

\begin{theorem}\label{thm:totally-ordered-approx}
For any unit-demand buyer with totally-ordered value for all items, item pricing gives a 5.4 approximation in revenue to the optimal buy-many mechanism.
\end{theorem}

\begin{proof}
Let $p$ be the optimal buy-many mechanism. Define $q$ to be the following item pricing: $q_i$, the price of item $i$, is the cheapest price at which an adaptive buyer can obtain an item with index at least $i$ with probability 1 from repeatedly purchasing a single lottery. In other words, 
\begin{equation*}
q_i=\min_{x\in \Delta([n]): \sum_{j\geq i}x_j=1}p(x).
\end{equation*}
We will show that a scaling of $q$ gives a constant fraction of the optimal revenue obtained by buy-many mechanisms. We use the following lemma from \cite{chawla2019buy} that relates the revenue of an appropriate scaling of $q$ to the change in the buyer's utility as the pricing function changes from a low scaling factor, $\ell$, to a high one, $h$.

\begin{lemma}\label{lem:CTT19_lem31}(Lemma 3.1 of \cite{chawla2019buy})
For any pricing $q$ and any $0<\ell\leq h$, let $\alpha$ be drawn from $[\ell,h]$ with density function $\frac{1}{\alpha\log(h/\ell)}$. Then for any valuation function $v$,
\begin{equation*}
\E_{\alpha}[\rev_{\alpha q}(v)]\geq\frac{u_{\ell q}(v)-u_{hq}(v)}{\ln(h/\ell)}.
\end{equation*}
\end{lemma}
In order to utilize this lemma, choosing $h=1$, we show that the buyer obtains a low utility under pricing $q$ and high enough utility at an appropriate scaling $\ell q$. We begin with two observations. 
% To prove this, we first have the following two observations.
The first shows that in any buy-many mechanism, the value of any set obtained by the buyer with non-zero probability is at least the expected utility of the buyer.
\begin{lemma}\label{lem:util-lb}
For any buyer type $v$ and any buy-many mechanism $p$, the buyer purchases $\lambda=(x,p(x))$ in $p$, then for any set $T$ in the support of $x$, $v(T)\geq u_p(v)$. 
\end{lemma}
\begin{proof}
Since $\lambda$ is the optimal lottery purchased by the buyer, purchasing it is also the optimal adaptive strategy of the buyer. Thus, if the buyer buys $\lambda$ and gets any set $T$ in the support of $x$ allocated, she would not purchase another lottery on the menu, in particular, $\lambda$. Since the value gain of purchasing $\lambda$ with set $T$ at hand is at most $v(\lambda)-v(T)$, therefore $v(\lambda)-v(T)\leq p(\lambda)$ which is the price of purchasing $\lambda$. Thus $v(T)\geq v(\lambda)-p(\lambda)=u_p(v)$.
\end{proof}
We emphasize that the lemma holds for arbitrary buyer types and not just unit-demand valuations. 
%Actually in the above observation, since the buyer is unit-demand in the totally ordered setting, any set $T$ in the support contains only a single item. However, it can also be applied to general-valued buyers.
The second observation shows that for any lottery, its price in $p$ is lower bounded by the price of some item in its support in item pricing $q$.
\begin{lemma}\label{lem:totally-ordered-lem}
For any allocation $x\in\Delta([n])$, there exists an item $i$ in the support of $x$, such that $q_i\leq p(x)$.
\end{lemma}
\begin{proof}
Let $i$ be the item with the lowest type in the support of $x$. Then we have $\sum_{j\geq i}x_j=1$, thus $q_i\leq p(x)$ by definition of $q$.
\end{proof}

Now we come back to the proof of the theorem. Fix any buyer type $v$. We will consider four different pricing mechanisms: the optimal buy-many pricing $p$, the item pricing $q$ constructed above, and their scalings $\beta p$ and $\ell q$ with $\beta,\ell\in[0,1]$. 
% Consider the following four mechanisms. Under the optimal pricing $p$, the buyer purchases lottery $\lambda$. Suppose that we scale the price of every lottery in $p$ down by a factor of $\beta$, and under pricing $\beta p$, the buyer purchases lottery $\lambda'$. Under item pricing $q$, the buyer purchases item $i$. Suppose that we scale the price of each item in $q$ down by a factor of $\ell$, and under pricing $\ell q$, the buyer purchases item $i'$. Here $\beta,\ell\in[0,1]$ are some to-be-determined constants.

Let $\lambda'$ denote the lottery the buyer purchases under pricing $\beta p$. By Lemma~\ref{lem:totally-ordered-lem}, there exists an item $j$ in the support of $\lambda'$, such that $p(\lambda')\geq q_j$. Then
\begin{equation}\label{eqn:totally-ordered1}
\rev_{\beta p}(v)=\beta p(\lambda')\geq \beta q_j.
\end{equation}
Next, let $\lambda$ denote the lottery the buyer purchases under pricing $p$. Then, by Lemma~\ref{lem:util-lb}, we have
\begin{equation}\label{eqn:totally-ordered2}
v_j\geq u_{\beta p}(v)\geq v(\lambda)-\beta p(\lambda)=v(\lambda)-p(\lambda)+(1-\beta)p(\lambda)=u_{p}(v)+(1-\beta)\rev_p(v).
\end{equation}
where the second inequality follows by noting that the buyer has the option of purchasing $\lambda$ under $\beta p$. Next, we note that since the buyer has the option of purchasing item $j$ under pricing $\ell q$, we have
%a larger utility purchasing item $i'$ than $j$ under pricing $\ell q$, we have
\begin{equation}\label{eqn:totally-ordered3}
u_{\ell q}(v)\geq v_j-\ell q_j.
\end{equation}
Finally, by the definition of the pricing $q$, the buyer can purchase any set of items more cheaply under $q$ than under $p$. Therefore, 
\begin{equation}\label{eqn:totally-ordered4} 
u_p(v)\geq u_q(v).
\end{equation} 

% Since the buyer has a larger utility purchasing lottery $\lambda'$ than buying the lottery $\lambda_i$ that defines $q_i$ in $p$ (which has price $q_i$), and $\lambda_i$ allocates the buyer an item that is at least as good as $i$, we have
% \begin{equation}\label{eqn:totally-ordered4}
% u_p(v)\geq v(\lambda_i)-p(\lambda_i)\geq v_i-p(\lambda_i)=v_i-q_i=u_q(v).
% \end{equation}
By $\frac{\ell}{\beta}\eqref{eqn:totally-ordered1}+\eqref{eqn:totally-ordered2}+\eqref{eqn:totally-ordered3}+\eqref{eqn:totally-ordered4}$,
\begin{equation}\label{eqn:util-difference-totally-ordered}
u_{\ell q}(v)-u_{q}(v)\geq (1-\beta)\rev_p(v)-\frac{\ell}{\beta}\rev_{\beta p}(v).
\end{equation}

% Here is an important lemma that was first proposed in \cite{chawla2019buy}. For any pricing function $q$, the utility difference of two scalings of $q$ can be captured by the revenue of a scaling of it.
% \begin{lemma}\label{lem:CTT19_lem31}(Lemma 3.1 of \cite{chawla2019buy})
% For any pricing $q$ and any $0<\ell\leq h$, let $\alpha$ be drawn from $[\ell,h]$ with density function $\frac{1}{\alpha\log(h/\ell)}$. Then for any valuation function $v$,
% \begin{equation*}
% \E_{\alpha}[\rev_{\alpha q}(v)]\geq\frac{u_{\ell q}(v)-u_{hq}(v)}{\ln(h/\ell)}.
% \end{equation*}
% \end{lemma}

By applying Lemma~\ref{lem:CTT19_lem31} to \eqref{eqn:util-difference-totally-ordered}, there exists a random scaling factor $\alpha$, such that 
\begin{equation}\label{eqn:totally-ordered5}
\rev_{\alpha q}(v)\geq\frac{u_{\ell q}(v)-u_{q}(v)}{\ln (1/\ell)}\geq\frac{1}{\ln(1/\ell)}\left((1-\beta)\rev_p(v)-\frac{\ell}{\beta}\rev_{\beta p}(v)\right).
\end{equation}
Since $p$ is the optimal buy-many mechanism, it achieves higher revenue than $\ell p$, which means $\E_{v\sim D}[\rev_p(v)]\geq\E_{v\sim D}[\rev_{\ell p}(v)]$. Taking the expectation over $v$ on both sides of \eqref{eqn:totally-ordered5}, we have
\begin{eqnarray*}
\rev_{\alpha q}&=&\E_{v\sim D}[\rev_{\alpha q}(v)]\\
%&\geq&\E_{v\sim D}\left(\frac{1}{\ln(1/\ell)}\left((1-\beta)\rev_p(v)-\frac{\ell}{\beta}\rev_{\beta p}(v)\right)\right)\\
&\geq&\frac{1}{\ln(1/\ell)}\left((1-\beta)\rev_p-\frac{\ell}{\beta}\rev_{\beta p}\right)
\geq \frac{1}{\ln(1/\ell)}\left(1-\beta-\frac{\ell}{\beta}\right)\rev_p.
\end{eqnarray*} 
Take $\ell=0.03485$ and $\beta=0.18668$, we have $\rev_{\alpha q}\geq 0.18668 \rev_p$. Since $\alpha q$ is a (randomized) item pricing, thus there exists an item pricing that gives a constant $1/0.18668<5.4$-approximation to the revenue obtained by the optimal buy-many mechanism.
%%
% Calculation: https://www.wolframalpha.com/input/?i=maximize+%281-2*sqrt%28x%29%29%2Fln%281%2Fx%29+ 

\end{proof}

\subsection{Hardness of computing the optimal item pricing}
\label{sec:np-hardness}

%To complement the PTAS algorithm for finding an approximately optimal item pricing in the totally-ordered setting, 
In this section, we show that it's strongly NP-hard to compute the optimal revenue that can be obtained by item pricing mechanisms, which means that there exists no FPTAS for the problem unless \textsc{P}=\textsc{NP}. 
%Together with the PTAS algorithm we give a tight characterization of the complexity of the problem.
Let $\orderedip$ denote the following problem: For a unit-demand buyer with ordered valuation over $n$ items, let $D$ be the value distribution with support size $m$. Compute the optimal revenue obtained by item pricing.

\begin{theorem}\label{thm:totally-ordered-nphard}
$\orderedip$ is strongly NP-hard, even when each realized buyer has at most three distinct item values.
\end{theorem}

\begin{proof}[Proof Sketch]
We prove the theorem via a reduction from \textsc{Max-Cut}. 

For any \textsc{Max-Cut} instance with graph $G(V,E)$, let $n=|V|>180$ be large enough. Consider an instance of $\orderedip$ with $n+1$ items. For convenience, we assume that each node in $V$ also has an index in $[n]$. We want the following properties of the optimal item pricing for the instance:
\begin{enumerate}
	\item The optimal item pricing has integral item prices for each item;
	\item $p_{n+1}=6n$, and there is a set of buyer types purchasing item $n+1$ with realization probability $q_1=0.9$ that do not depend on the structure of the graph and contribute $R_1(n)$ to the total revenue;
	\item $p_{i}=3i-1$ or $3i-2$, and there is a set of buyer types purchasing items in $[n]$ with realization probability $q_2<\frac{1}{12n}q_1$ that do not depend on the structure of the graph and contribute $R_2(n)$ to the total revenue;
	\item In addition to all previous buyers, for each $(i,j)\in E$ with $i<j$, there exists a set $T_{ij}$ of buyer types, and real number $R_{ij}(n)>0$ that is irrelevant to the graph structure such that: if $p_j-p_i=3(j-i)$, then the revenue contribution from $T_{ij}$ is $R_{ij}(n)$; if $p_j-p_i\neq 3(j-i)$, then the revenue contribution from $T_{ij}$ is $R_{ij}(n)+\frac{1}{n^{10}}$. The realization probability of any buyer type is polynomially bounded by $n$ (at least $poly(n^{-1})$).
\end{enumerate}

The construction of the instance is deferred to Section~\ref{sec:np-hardness-proof}. Here we show the strongly NP-hardness of $\orderedip$ for an instance with above properties. This proves the claim of the Theorem.

Given an instance with such properties, we can calculate the revenue of the optimal item pricing for the instance. The total revenue contributed from buyer types from Property 2 and 3 is $R_1(n)+R_2(n)$. For any cut $C=(V_1,V\setminus V_1)$, if each item $i$ in $V_1$ is priced $3i-1$, while each item $j$ in $V\setminus V_1$ is priced $3j-2$, the total revenue contributed from buyer types from Property 4 is $\sum_{(i,j)\in E}R_{ij}(n)+\frac{1}{n^{10}}|C|$. Thus, for a graph $G(V,E)$ with maximum cut $c_{\max}$, the corresponding instance of $\orderedip$ has maximum revenue 
\begin{equation*}
h(G)=R_1(n)+R_2(n)+\sum_{(i,j)\in E}R_{ij}(n)+\frac{1}{n^{10}}c_{\max}.
\end{equation*} 
This builds a bijection between the maximum cut of $G$, and the optimal item pricing revenue of the $\orderedip$ instance constructed from $G$. Since all inputs for the $\orderedip$ instance are polynomially bounded, we know that $\orderedip$ is strongly NP-hard from the APX-hardness of \textsc{Max-Cut} (see Eg. \cite{haastad2001some,trevisan2000gadgets}).

\end{proof}

\subsection{A PTAS for computing a near-optimal item pricing}
\label{sec:totally-ordered-ptas}

We complement the strongly NP-hardness result by developing a PTAS algorithm for computing the optimal item pricing for a unit-demand buyer with totally-ordered values. We present a detailed proof sketch here and complete details are deferred to Appendix~\ref{sec:proof-totally-ordered-ptas}.

\begin{theorem}\label{thm:totally-ordered-algorithm}
For a unit-demand buyer with totally ordered item values, there exists an algorithm running in $poly(m,n^{poly(1/\eps)},\log\range)$ time that computes an item pricing that is $(1+\eps)$-approximation in revenue to the optimal item pricing.\footnote{The dependency on $\log\range$ in the running time can actually be removed due to Lemma~\ref{lem: support-size}.}
\end{theorem}

\begin{proof}[Proof Sketch]

We prove the theorem in several steps.
\begin{enumerate}
\item There exists a near optimal item pricing where all prices are non-decreasing powers of $(1+\eps^2)$. Specifically let $\Pi=\{(1+\eps^2)^r | r\in\Z\}\cup\{0\}$. Then for all item pricings $p$, there exists $q^{(1)}\in\Pi^n$, such that for all value functions $v$,
  \begin{align*}
    \rev_{q^{(1)}}(v) \ge (1-O(\eps))\rev_p(v). 
  \end{align*}
\item At a small loss in revenue, we can restrict prices to lie in a small set. In particular, for all value distributions $\dist$ with value range $\range$, there exists an efficiently computable set $\Pi^*\subset\Pi$ with $|\Pi^*|=\poly(1/\eps,\log \range)$ such that for all item pricings $q^{(1)}\in \Pi^n$, there exists an item pricing $q^{(2)}\in {\Pi^*}^n$ satisfying
  \begin{align*}
    \rev_{q^{(2)}}(\dist) \ge (1-O(\eps))\rev_{q^{(1)}}(\dist). 
  \end{align*}
\item Next, we define for each unit demand buyer an additive value function that closely mimics it. For a unit-demand value function $v$, define $\vadd_i=v_i-v_{i-1}$ and let $\vadd$ be the value function that assigns to a set $S\subseteq [n]$ the value $\sum_{i\in S} \vadd_i$.

  We also define a new kind of pricing that we will call an {\em interval prefix pricing}. Given a partition of the $n$ items into $t$ intervals, $I_{i_0,i_1}$, $I_{i_1,i_2}$, $\cdots$, $I_{i_{t-1},i_{t}}$ with $i_0=0$ and $i_t=n$, an interval prefix pricing $q$ is a menu with $n$ options; the $j$th option allocates the set $I_{i_\ell,j}=\{i_\ell+1, i_\ell+2, \cdots, j\}$ at price $q_j$ where $i_\ell< j\le i_{\ell+1}$.

  We furthermore say that an interval prefix pricing $q$ satisfies {\em price gap} $(\gamma,\delta)$ if (1) menu options corresponding to different intervals are priced multiplicatively apart: for all $i$, $j$, and $\ell$ with $i\le i_\ell<j$, $q_j\ge (1+\eps^2)^\gamma q_i$, (2) and, menu options corresponding to any single interval are priced multiplicatively close to each other: for all $i$, $j$, and $\ell$ with $i_\ell<i<j\le i_{\ell+1}$, $q_j\le (1+\eps^2)^\delta q_i$.
  
  We show that for value distribution $\dist$ and its corresponding additive value distribution $\dist^{\oplus}$ and item pricing $q^{(2)}\in {\Pi^*}^n$, there exists an efficiently computable set $\Pi'$ with $|\Pi'|=|\Pi^*|$ and an interval prefix pricing $q^{(3)}$ with $q^{(3)}_i\in \Pi'$ for all $i\in [n]$ and price gap $(\frac{1}{\eps^2}\ln\frac{1}{\eps^2},\frac{1}{\eps^3}\ln\frac{1}{\eps^2})$, such that
  \begin{align*}
    \rev_{q^{(3)}}(\dist^{\oplus}) \ge (1-O(\eps))\rev_{q^{(2)}}(\dist).
  \end{align*}

	The converse also holds: for every unit demand value function $v$ and interval prefix pricing $q$ with price gap $(\frac{1}{\eps^2}\ln\frac{1}{\eps^2},\frac{1}{\eps^3}\ln\frac{1}{\eps^2})$, we can efficiently compute an item pricing $q^{(4)}$ such that
  \begin{align*}
    \rev_{q^{(4)}}(v) \ge (1-O(\eps)) \rev_{q}(\vadd).
  \end{align*}

\item Finally, we show that for any distribution over additive values $\vadd$ and any set $\Pi'\subset \Pi$ an optimal interval prefix pricing $q$, with $q_i\in \Pi'$ for all $i\in [n]$ and price gap $(\frac{1}{\eps^2}\ln\frac{1}{\eps^2},\frac{1}{\eps^3}\ln\frac{1}{\eps^2})$, can be found in time polynomial in $|\Pi'|$, $n^{poly(1/\eps)}$, and $m$.
  
\end{enumerate}

The algorithm can now be described as follows. Given the distribution $\dist$ we first compute the additive valuation $\vadd$ for every buyer type $v$ in the support of $\dist$. We also compute the set $\Pi'$ of allowable prices. By Step 4, we can efficiently compute the optimal interval prefix pricing $q$ with price gap $(\frac{1}{\eps^2}\ln\frac{1}{\eps^2},\frac{1}{\eps^3}\ln\frac{1}{\eps^2})$ for the distribution $\dist^{\oplus}$ over the additive values $\vadd$, such that all item prices are in $\Pi'$. Then, by Step 3, we can efficiently compute an item pricing $q^{(4)}$ with $\rev_{q^{(4)}}(v) \ge (1-O(\eps)) \rev_{q}(\vadd)$. We return the pricing $q^{(4)}$, which satisfies:
\begin{eqnarray*}
\rev_{q^{(4)}}(\dist)&\geq&(1-O(\eps))\rev_{q}(\dist^{\oplus})\geq(1-O(\eps))\rev_{q^{(3)}}(\dist^{\oplus})\geq (1-O(\eps))\rev_{q^{(2)}}(\dist)\\
&\geq& (1-O(\eps))\rev_{q^{(1)}}(\dist)\geq (1-O(\eps))\rev_{p}(\dist)
\end{eqnarray*}
where $p$ is the optimal item pricing and $q$ is the optimal interval prefix pricing for $\dist^{\oplus}$. The four inequalities follow from Step 3, Step 3, Step 2, Step 1 respectively. $q^{(4)}$ can be found in $poly(m,n^{poly(1/\eps)},|\Pi'|)=poly(m,\log\range,n^{poly(1/\eps)})$ time. 

Now we briefly elaborate on each step. A more detailed discussion is deferred to Appendix~\ref{sec:proof-totally-ordered-ptas}.

The first two steps are standard. Step 1 follows from Lemma~\ref{lem:Nisan} and the fact that any item pricing is $(1+\eps^2)$-approximated by a power-of-$(1+\eps^2)$ pricing. Step 2 follows from the fact that when the positive item values are bounded in range $[1,\range]$, then there exists an item pricing with all item prices in range $[\eps^2,\range]$ with $(1-O(\eps))$-fraction of the optimal revenue. 

To prove Step 3, we first show that there exists an interval partition of the items and a corresponding item pricing $q$ with price gap $(\gamma,\delta)=(\frac{1}{\eps^2}\ln\frac{1}{\eps^2},\frac{1}{\eps^3}\ln\frac{1}{\eps^2})$, that achieves $(1-O(\eps))$-fraction of the optimal item pricing revenue. For any unit-demand buyer $v$, this is equivalent to a prefix-pricing $q'$ that sells prefixes $\{1,2,\cdots,i\}$ at price $q_i$. For the corresponding additive buyer $\vadd$ of $v$, they have the same behavior under $q'$. Also observe for the interval prefix pricing $q''$ that is defined by $q''_i=q_i$, the price of any set $S$ under $q''$ is almost the same as the price under $q'$, since the item prices from different intervals are at least $(1+\eps^2)^{\gamma}>\frac{1}{\eps^2}$ apart; the predominant part of $q''(S)$ comes from the interval $I_{\ell}$ with the largest index that intersects $S$, which is exactly $q'(S)$. Thus $q''(S)$ and $q'(S)$ are very close to each other, which means a scaling of $q''$ achieves a revenue close to that of $q'$, and vice versa.

Step 4 can be solved via a dynamic program. To compute the optimal interval prefix pricing with price gap $(\gamma,\delta)=(\frac{1}{\eps^2}\ln\frac{1}{\eps^2},\frac{1}{\eps^3}\ln\frac{1}{\eps^2})$ for $\vadd$ with all item prices in set $\Pi'$, since the buyer is additive, the revenue contribution from each interval can be calculated separately, without worrying about the incentive of the buyer. In each interval, since the prices are only $(1+\eps^2)^{\delta}$ apart, there are only $n^{\poly(1/\eps)}$ different item prices in an interval, and we can enumerate all pricings within the interval efficiently. We only need to make sure that the prices in all intervals are monotone and satisfy the $(\gamma,\delta)$ price gap, and this can be done via a dynamic program. 

\end{proof}

\subsection{Discussion on additive buyers}
\label{sec:totally-ordered-additive}

In Section~\ref{sec:totally-ordered-approx}, we showed that for a unit-demand buyer with totally ordered item values, item pricing gives a constant approximation to the optimal buy-many revenue. We want to investigate whether such nice properties generalizes to other class of valuation functions. For example, for an additive buyer with totally-ordered item values, can item pricings achieve a constant fraction of the optimal buy-many revenue? Unfortunately, we give a negative answer to the question through the following theorem.

\begin{theorem}\label{thm:totally-ordered-approx-additive}
	For any additive buyer with totally ordered value for all items, item pricing cannot achieve an approximation ratio $o(\log\log n)$ to the optimal (deterministic) buy-many mechanism.
\end{theorem} 
The proof uses a reduction from an additive buyer with unordered item values, which is known to have $\Omega(\log n)$ revenue gap between the optimal item pricing and the optimal deterministic buy-many mechanism \cite{chawla2019buy}. See Section~\ref{sec:proof-totally-ordered-additive} for the complete proof.
\section{Item Pricing in the Partially-Ordered Setting}
\label{sec:partially-ordered}

In this section, we generalize the results for a unit-demand buyer in the totally-ordered setting to a general-valued buyer in the partially-ordered setting. We first show that item pricing gives an $O(\log k)$-approximation in revenue to the optimal buy-many mechanism, where $k$ is the width of the ordered set of items. Then we provide a PTAS algorithm for finding a near-optimal item pricing when $k$ is a constant. This way, we show that the width $k$ of the ordered set is the key parameter in both the performance and the computational complexity of the item pricing mechanisms.

\subsection{Item pricing gives an \texorpdfstring{$O(\log k)$}{} approximation in revenue to the optimal buy-many mechanism}
\label{sec:partially-ordered-logk}

For a general-valued buyer, \cite{chawla2019buy} shows that item pricing gives an $O(\log n)$ approximation in revenue to the optimal buy-many mechanism. Here we improve the approximation ratio to $O(\log k)$, which is only related to the width of the partially ordered item set.

\begin{theorem}\label{thm:k-width-general-approx}
For any buyer with partially ordered values for all items with width $k$, item pricing gives an $O(\log k)$ approximation in revenue to the optimal buy-many mechanism. 
\end{theorem}

Similar to the totally-ordered setting, given an optimal mechanism $p$, we define item pricing $q$ such that each $q_i$ is the cheapest way to get an item that dominates $i$ in $p$. We show that a scaling of $q$ gives $O(\log k)$-approximation to the optimal revenue. Compared to the totally-ordered setting, the major difference is that Lemma~\ref{lem:totally-ordered-lem} no longer holds. Instead, we show that for any allocation $x$ of items, there exists a set $T$ in the support of $x$, such that $q(T)\leq k^2p(x)$. The rest of the proof mostly still goes through.

\begin{proof}
%The proof follows the same flow as Theorem~\ref{thm:totally-ordered-approx}, and we only need to change the definition of the benchmark item prices and the proof of the lemmas. 
Let $\feasiblesets$ be the collection of sets that can be demanded by the buyer. Given the optimal buy-many mechanism $p$, define item pricing $q$ as follows: $q_{i}$, the price of item $i$, is defined by the cheapest way of getting an item that dominates it. In other words, if $S_i=\{j\in[n]|j\succeq i\}$ is the set of items that dominate $i$,
\begin{equation*}
q_i=\min_{x\in \Delta(\feasiblesets): \sum_{T\cap S_i\neq \emptyset}x_T=1}p(x)=\min_{x\in \Delta(\feasiblesets)}\frac{p(x)}{\sum_{T\cap S_i\neq \emptyset}x_T}.
\end{equation*}
We show that a scaling of $q$ gives a $1/O(\log k)$ fraction of the optimal revenue obtained by buy-many mechanisms. Lemma~\ref{lem:util-lb} still holds, while Lemma~\ref{lem:totally-ordered-lem} needs to be modified as follows.
\begin{lemma}\label{lem:k-width-general-lem}
For any allocation $x\in\Delta(\feasiblesets)$, there exists a set $T$ in the support of $x$, such that $q(T)\leq k^2 p(x)$.
\end{lemma}
\begin{proof}
Let $S$ be the set of items, such that for any $i\in S$, the probability of getting an item that dominates $i$ is less than $\frac{1}{k}$. In other words, 
\begin{equation*}
S=\left\{i\in[n]\Big|\sum_{T\cap S_i\neq \emptyset}x_T<\frac{1}{k}\right\}.
\end{equation*}
We first observe that there exists a set $T$ in the support of $x$, such that $T$ does not contain any item in $S$. We reason this as follows. Let $S'\subseteq S$ be an antichain such that any element in $S$ dominates an element in $S'$. Intuitively, $S'$ is the ``bottom'' of set $S$ in the preference graph. Then for any $i\in S'$, the probability that a set drawn from $x$ contains an element that dominates $i$ is less than $\frac{1}{k}$. Since $S'$ is an antichain and there are at most $k$ elements in $S'$, by union bound, the probability that a set drawn from $x$ contains an element that dominates \textit{some} element in $S'$ is less than 1. Thus there exists a set $T$ in the support of $x$, such that $T$ does not contain any item that dominates some element in $S'$, thus does not contain any item in $S$.

For any item $i\in T$, by the definition of $q_i$,
\begin{equation*}
q_i\leq\frac{p(x)}{\sum_{T\cap S_i\neq \emptyset}x_T}\leq\frac{p(x)}{1/k}=kp(x).
\end{equation*}
Then 
\begin{equation*}
q(T)=\sum_{i\in T}q_i\leq \sum_{i\in T}kp(x)=|T|kp(x)\leq k^2p(x).
\end{equation*}

\end{proof}

Now we are ready to prove the theorem. The same as in the proof of Theorem~\ref{thm:totally-ordered-approx}, we fix any buyer type $v$, and define the four mechanisms as follows. Under the optimal pricing $p$, the buyer purchases lottery $\lambda$; under pricing $\beta p$ the buyer purchases lottery $\lambda'$; under item pricing $q$, the buyer purchases set $S$; under pricing $\ell q$ the buyer purchases set $S'$. Here we define $\beta=\frac{1}{2}$ and $\ell=\frac{1}{8k^2}$.

By Lemma~\ref{lem:k-width-general-lem}, there exists a set $T$ in the support of $\lambda'$, such that $p(\lambda')\geq \frac{1}{k^2}q(T)$. Then
\begin{equation}\label{eqn:k-width-general1}
\rev_{\beta p}(v)=\beta p(\lambda')\geq \frac{\beta}{k^2} q(T).
\end{equation}
Inequality \eqref{eqn:totally-ordered2} still holds (by replacing item $j$ in \eqref{eqn:totally-ordered2} with set $T$):
\begin{equation}\label{eqn:k-width-general2}
v(T)\geq u_{p}(v)+(1-\beta)\rev_p(v).
\end{equation}
Since the buyer has a larger utility purchasing item $S'$ than $T$ under pricing $\ell q$, we have
\begin{equation}\label{eqn:k-width-general3}
u_{\ell q}(v)\geq v(T)-\ell q(T).
\end{equation}
Since the buyer has a larger utility purchasing lottery $\lambda$ than buying the collection of lottery $\lambda_i$ for every $i\in S$ that defines $q_{i}$ in $p$ (which has price $q_{i}$), and $\lambda_i$ allocates the buyer an item that dominates $i$, we have 
\begin{equation}\label{eqn:k-width-general4}
u_p(v)\geq v(\cup \lambda_i)-\sum_{i\in S}p(\lambda_i)\geq v(S)-\sum_{i\in S}q_i=v(S)-q(S)=u_q(v).
\end{equation}
By $\frac{\ell k^2}{\beta}\eqref{eqn:k-width-general1}+\eqref{eqn:k-width-general2}+\eqref{eqn:k-width-general3}+\eqref{eqn:k-width-general4}$, and apply $\beta=\frac{1}{2}$ and $\ell=\frac{1}{8k^2}$,
\begin{equation}\label{eqn:util-difference}
u_{\ell q}(v)-u_{q}(v)\geq (1-\beta)\rev_p(v)-\frac{\ell k^2}{\beta}\rev_{\beta p}(v)=\frac{1}{2}\rev_p(v)-\frac{1}{4}\rev_{\beta p}(v).
\end{equation}

By applying Lemma~\ref{lem:CTT19_lem31} to \eqref{eqn:util-difference}, there exists a random scaling factor $\alpha$, such that 
\begin{equation}\label{eqn:k-width-general5}
\rev_{\alpha q}(v)\geq\frac{u_{\ell q}(v)-u_{q}(v)}{\ln (1/\ell)}\geq\frac{1}{\ln (8k^2)}\left(\frac{1}{2}\rev_p(v)-\frac{1}{4}\rev_{\beta p}(v)\right).
\end{equation}
Since $p$ is the optimal buy-many mechanism, it achieves higher revenue than $\ell p$, which means $\E_{v\sim D}[\rev_p(v)]\geq\E_{v\sim D}[\rev_{\ell p}(v)]$. Taking the expectation over $v$ on both sides of \eqref{eqn:k-width-general5}, the same as \eqref{eqn:totally-ordered5} we have
\begin{equation*}
\rev_{\alpha q}\geq\frac{1}{\ln(8k^2)}\left(\frac{1}{2}\rev_p-\frac{1}{4}\rev_{\beta p}\right)=\frac{1}{4\ln(8k^2)}\rev_p.
\end{equation*} 
Thus there exists an item pricing which gives a $\frac{1}{O(\log k)}$ fraction of the revenue obtained by the optimal buy-many mechanism.

\end{proof}

\subsection{A PTAS algorithm for computing a near-optimal item pricing in partially-ordered setting}
\label{sec:partially-ordered-ptas}

In this section, we generalize the approximation algorithm for the totally-ordered setting to the partially ordered setting, where the width of the entire set of items is $k$. When $k$ is a constant, the algorithm is a PTAS.

\begin{theorem}\label{thm:general-value-approx-algorithm}
For a general-valued buyer with partially ordered values, if the partially ordered set containing all items has width $k$, then there exists an algorithm running in $poly(m,n^{poly(k,1/\eps)},\log\range)$ time that computes an item pricing that is $(1+\eps)$-approximation in revenue to the optimal item pricing.
\end{theorem}

\begin{proof}[Proof Sketch]

The algorithm is similar to the totally-ordered setting, but we need to define the generalized multi-dimensional \textit{prefixes} and \textit{intervals}. We also define slightly different \textit{interval prefix pricing} and \textit{additive buyer}.
\begin{itemize}
	\item \textit{Prefix}: For any set $T\in\feasiblesets$, the prefix parameterized by $T$ is a set of items $P_T\supseteq T$ such that item $j\in[n]$ is in $P_T$ if and only if there exists $i\in T$ such that $j\preceq i$. Such a definition ensures that for $T_1,T_2\in\feasiblesets$,  $T_1\preceq T_2$, $P_{T_1}\subseteq P_{T_2}$.

	\item \textit{Interval}: For two sets $T,T'\in\feasiblesets$ with $T\subseteq T'$, interval $I_{T,T'}=P_{T'}\setminus P_{T}$ is a set of items with contiguous item types between $T$ and $T'$ in the ordering graph. 

	\item \textit{Interval prefix pricing}: Let $I=(I_{T_0,T_1},I_{T_1,T_2},\cdots,I_{T_{t-1},T_{t}})$ be  a partition of the $n$ items into $t$ intervals,  with $T_0=\emptyset$, and $T_t$ be a set of items that dominates all other items (with $P_{T_t}=[n]$). An interval prefix pricing $q$ is a mechanism defined by a vector of item prices $(q_1,q_2,\cdots,q_n)$: For any set $S_{\ell}\subseteq I_{\ell}$ of items, there is a menu allocating a set of items $P^*_{S_{\ell}}=\cup_{1\leq j\leq \ell-1}T_{j}\cup S_{\ell}$, with price $q(P^*_{S_{\ell}})=\sum_{i\in P^*_{S_{\ell}}}q_i$. In other words, to purchase any set of items in $S_{\ell}$, the buyer also needs to purchase all sets of items $T_{1},\cdots,T_{\ell-1}$ that define the previous intervals.

	\item \textit{Additive-over-intervals buyer}: Given an interval partition $I=(I_{T_0,T_1},I_{T_1,T_2},\cdots,I_{T_{t-1},T_{t}})$, for an arbitrary value function $v$ and any set $S=S_1\cup S_2\cup \cdots\cup S_t$ of items with $S_\ell\subseteq I_\ell$ for every $\ell\in[t]$, define
		\begin{equation*}
			\vaddi(S)=\sum_{\ell=1}^{t}\big(v(S_\ell\cup T_{\ell-1})-v(T_{\ell-1})\big).
		\end{equation*}
	In other words, $\vaddi(S_{\ell})$ is the value gain of getting set $S_{\ell}$, when the buyer $v$ has a set of items $T_{\ell-1}$ at hand. $v$ and $\vaddi$ has the same behavior under an interval prefix pricing defined by interval partition $I$. 
\end{itemize}

The proof for the totally-ordered setting can be generalized to the partially ordered setting, if we use the above generalized definitions of the terms. The key steps are shown as follows.

\begin{enumerate}
\item There exists a near optimal item pricing where all prices are powers of $(1+\eps^2)$. Let $\Pi=\{(1+\eps^2)^r | r\in\Z\}\cup\{0\}$. Then for all item pricings $p$, there exists $q^{(1)}\in\Pi^n$, such that for all value functions $v$,
  \begin{align*}
    \rev_{q^{(1)}}(v) \ge (1-O(\eps))\rev_p(v). 
  \end{align*}
Furthermore, without loss of generality, we can assume that for the item pricing $q^{(1)}$ we consider, set $\{i|q^{(1)}_i\leq y\}$ is a prefix for any $y\in\R$.
\item At a small loss in revenue, we can restrict prices to lie in a small set. In particular, for all value distributions $\dist$ with value range $\range$, there exists an efficiently computable set $\Pi^*\subset\Pi$ with $|\Pi^*|=poly(1/\eps,k,\log n, \log \range)$ such that for all item pricings $q^{(1)}\in \Pi^n$, there exists an item pricing $q^{(2)}\in {\Pi^*}^n$ satisfying
  \begin{align*}
    \rev_{q^{(2)}}(\dist) \ge (1-O(\eps))\rev_{q^{(1)}}(\dist). 
  \end{align*}
\item Given a partition of the $n$ items into $t$ intervals, $I=(I_{T_0,T_1},I_{T_1,T_2},\cdots,I_{T_{t-1},T_{t}})$ with $T_0=\emptyset$, and $T_t$ be a set of items that dominates all other items (with $P_{T_t}=[n]$). We furthermore say that for any interval partition $I$, an item pricing $q$ satisfies {\em price gap} $(\gamma,\delta)$ if (1) items corresponding to different intervals are priced multiplicatively apart: for all $i$, $j$, and $\ell$ with $i\in P_{T_{\ell}}$ and $j\not\in P_{T_{\ell}}$, $q_j\ge (1+\eps^2)^\gamma q_i$, (2) and, menu options corresponding to any single interval are priced multiplicatively close to each other: for all $i$, $j$, and $\ell$ with $i,j\in I_{\ell}$, $q_j\le (1+\eps^2)^\delta q_i$.

We show that for value distribution $\dist$ and item pricing $q^{(2)}\in {\Pi^*}^n$, there exists an item pricing $q^{(3)}$ with $q^{(3)}_i\in \Pi^*$ for all $i\in [n]$ and price gap $(\gamma,\delta)=(\frac{1}{\eps^2}\ln\frac{k}{\eps^2},\frac{k}{\eps^3}\ln\frac{k}{\eps^2})$, such that
  \begin{align*}
    \rev_{q^{(3)}}(\dist) \ge (1-O(\eps))\rev_{q^{(2)}}(\dist).
  \end{align*}

\item We define a new kind of pricing that we will call an {\em interval prefix pricing}. Given a partition of the $n$ items into $t$ intervals $I=(I_{T_0,T_1},I_{T_1,T_2},\cdots,I_{T_{t-1},T_{t}})$, an interval prefix pricing $q$ is a mechanism defined by a vector of item prices $(q_1,q_2,\cdots,q_n)$: For any set $S_{\ell}\subseteq I_{\ell}$ of items, there is a menu allocating a set of items $P^*_{S_{\ell}}=\cup_{1\leq j\leq \ell-1}T_{j}\cup S_{\ell}$, with price $q(P^*_{S_{\ell}})=\sum_{i\in P^*_{S_{\ell}}}q_i$. In other words, to purchase any set of items in $S_{\ell}$, the buyer also needs to purchase all sets of items $T_{1},\cdots,T_{\ell-1}$ that define the previous intervals.

We show that for every value function $v$ and item pricing $q^{(3)}\in {\Pi^*}^n$ with price gap $(\gamma,\delta)=(\frac{1}{\eps^2}\ln\frac{k}{\eps^2},\frac{k}{\eps^3}\ln\frac{k}{\eps^2})$, there exists an efficiently computable set $\Pi'$ with $|\Pi'|=|\Pi^*|$ and an interval prefix pricing $q^{(4)}$ with $q^{(4)}_i\in \Pi'$ for all $i\in [n]$ and price gap $(\frac{1}{\eps^2}\ln\frac{k}{\eps^2},\frac{k}{\eps^3}\ln\frac{k}{\eps^2})$, such that
  \begin{align*}
    \rev_{q^{(4)}}(v) \ge (1-O(\eps))\rev_{q^{(3)}}(v).
  \end{align*}

The converse is also true: for every value function $v$ and interval prefix pricing $q$ with price gap $(\frac{1}{\eps^2}\ln\frac{k}{\eps^2},\frac{k}{\eps^3}\ln\frac{k}{\eps^2})$, we can efficiently compute an item pricing $q^{(5)}$ such that
  \begin{align*}
    \rev_{q^{(5)}}(v) \ge (1-O(\eps)) \rev_{q}(v).
  \end{align*}

\item We define for each arbitrary-valued buyer $v$ and interval partitioning $I$ an \textit{additive-over-intervals} value function that closely mimics it. For an arbitrary value function $v$, for any set $S=S_1\cup S_2\cup \cdots\cup S_t$ of items with $S_\ell\subseteq I_\ell$ for every $\ell\in[t]$, define
\begin{equation*}
	\vaddi(S)=\sum_{\ell=1}^{t}\big(v(S_\ell\cup T_{\ell-1})-v(T_{\ell-1})\big).
\end{equation*}
In other words, $\vaddi(S_{\ell})$ is the value gain of getting set $S_{\ell}$, when the buyer $v$ has set of items $T_{\ell-1}$ at hand. We write $\dist_I^{\oplus}$ as the distribution of $\vaddi$ corresponding to $v\sim \dist$. $v$ and $\vaddi$ have the same behavior under an interval prefix pricing defined by interval partition $I$. When the interval partition is clear from the context, we will omit $I$ and write $\vadd=\vaddi$.

We show that for every additive-over-intervals value function $\vadd$ and interval prefix pricing $q^{(4)}$ with $q^{(4)}_i\in \Pi'$, and with price gap $(\gamma,\delta)=(\frac{1}{\eps^2}\ln\frac{k}{\eps^2},\frac{k}{\eps^3}\ln\frac{k}{\eps^2})$, there exists an efficiently computable set $\Pi^o$ with $|\Pi^o|=|\Pi^*|$ and an item $q^{(6)}$ with $q^{(6)}_i\in \Pi^o$ for all $i\in [n]$ and price gap $(\frac{1}{\eps^2}\ln\frac{k}{\eps^2},\frac{k}{\eps^3}\ln\frac{k}{\eps^2})$, such that
  \begin{align*}
    \rev_{q^{(6)}}(\vadd) \ge (1-O(\eps))\rev_{q^{(4)}}(v).
  \end{align*}

The converse also holds: for every value function $v$ and item pricing $q$ with price gap $(\frac{1}{\eps^2}\ln\frac{k}{\eps^2},\frac{k}{\eps^3}\ln\frac{k}{\eps^2})$, we can efficiently compute an interval prefix pricing $q^{(7)}$ with price gap $(\frac{1}{\eps^2}\ln\frac{k}{\eps^2},\frac{k}{\eps^3}\ln\frac{k}{\eps^2})$ such that
  \begin{align*}
    \rev_{q^{(7)}}(v) \ge (1-O(\eps)) \rev_{q}(\vadd).
  \end{align*}

\item Finally, we show that for any distribution $\dist$ over arbitrary values and any set $\Pi^o$ of values, an optimal item pricing $q$ for value distribution $\dist_I^{\oplus}$ and the corresponding interval partition $I$, with $q_i\in \Pi^o$ for all $i\in [n]$ and price gap $(\frac{1}{\eps^2}\ln\frac{k}{\eps^2},\frac{k}{\eps^3}\ln\frac{k}{\eps^2})$, can be found in time polynomial in $|\Pi^o|$, $n^{k\cdot poly(1/\eps)}$, and $m$.
  
\end{enumerate}

The reasoning is similar to the proof of Theorem~\ref{thm:totally-ordered-algorithm}. We defer the complete proof to Section~\ref{sec:partially-ordered-appendix}.

\end{proof}

We further notice that for a unit-demand buyer, the dependency on $\log\range$ can be removed. This is enabled via the following lemma.

\begin{lemma}\label{lem: support-size}
Let $V$ be a set of non-negative real numbers. Then we can efficiently find a set of power-of-$(1+\eps^2)$ prices $V'$ with $|V'|=O(|V|^2\frac{1}{\eps^2}\ln\frac{1}{\eps^2})$ satisfying the following: For any buyer that is unit-demand over $n$ items such that for any buyer type $v$ and item $i$ the buyer has value $v_{i}\in V$, the optimal power-of-$(1+\eps^2)$ item pricing $p$ satisfy $p_i\in V'$ for any $i\in[n]$. 
\end{lemma}

Given the lemma, since there are at most $mn$ different item values in the input, we have $|V|=mn$ in the lemma. The lemma can replace Step 2 in the proof of Theorem~\ref{thm:totally-ordered-algorithm} to remove the running time's dependency on $\log \range$. It can also be applied to Theorem~\ref{thm:general-value-approx-algorithm} for a unit-demand buyer. The proof is deferred to Section~\ref{sec:partially-ordered-appendix}.
% \input{approximation}
% \input{computation}

%\section*{Acknowledgement}

% Bibliography
\bibliographystyle{plainnat}
\bibliography{reference}

%\newpage
%\appendix
\appendix
%\section{Omitted Proofs in Section~\ref{sec:fedex}}
\section{Omitted Proofs for the FedEx Setting}
\label{sec:fedex-appendix}

\begin{numberedlemma}{\ref{lem:two-item-value-lemma}}
There exists an optimal item pricing such that $p_1\in \Pi_L$, and $p_i\in \Pi^*$ for each $i\geq 2$.
\end{numberedlemma}
\begin{proof}
We start with an arbitrary optimal item pricing $p$ with monotone item prices. Notice that a buyer $v$ will either purchase item 1 or item $i_v$ due to the monotonicity of item prices.

If $p_1\not\in \Pi_L$, and $p_1=p_2=\cdots=p_\ell<p_{\ell+1}$, suppose that we raise the price of item $1,2,\cdots,\ell$ by a small enough $\eps$. For a buyer of type $v$, if the buyer prefers to purchase item $1$ previously, then $v_1\geq p_1$. Since $p_1\not\in \Pi_L$, and $v_1\in \Pi_L$, we know that $v_1>p_1$, thus $v_1\geq p_1+\eps$ for a small enough $\eps$. Thus after the perturbation, she either still prefers to purchase item 1 and pay $\eps$ more, or she will switch to purchase a more expensive item, which means that her payment increases. For any buyer that does not prefer to purchase item $1$ before perturbation, her preferred item does not change after the price changes. Therefore, after the operation, the total revenue does not decrease. The value $\eps$ can be chosen so as to enforce either $p_1+\eps\in \Pi_L$, or $p_1+\eps=p_{\ell+1}$. By repeating this operation, we will have $p_1=p_2=\cdots=p_j\in \Pi_L$ for some $j\in[n]$, while maintaining the optimality of $p$.

If $\ell$ is the smallest index such that $p_{\ell}\not\in \Pi^*$, and $p_{\ell}=p_{\ell+1}=\cdots=p_{j}<p_{j+1}$, suppose that we raise the price of items $\ell,\ell+1,\cdots,j$ by a small enough $\eps$. For a buyer of type $v$, if $i_v\not\in[\ell,j]$, the buyer's incentive does not change after the perturbation. Otherwise when $\ell\leq i_v\leq j$, if the buyer prefers to purchase item $i_v$ before the perturbation, then $v_{i_v}-p_{i_v}\geq v_1-p_1$, and $v_{i_v}\geq p_{i_v}$. Since $p_1\in \Pi_L$, $v_{i_v}-v_1+p_1$ and $v_{i_v}$ are both in $\Pi^*$. By $p_{i_v}=p_{\ell}\not\in \Pi^*$, $v_{i_v}-p_{i_v}> v_1-p_1$ and  $v_{i_v}>p_{i_v}$. Therefore after the perturbation, for small enough $\eps$, she still prefers to purchase item $i_v$ and pay $\eps$ more. 
%If the buyer does not prefer to purchase item $\ell$ before the perturbation, her preferred item does not change after the price changes. 
Therefore, after the operation, the total revenue does not decrease. The value $\eps$ can be as large as making $p_\ell+\eps\in \Pi^*$, or making $p_\ell+\eps=p_{j+1}$. By repeating this operation, we will have $p_i\in \Pi^*$ for each $i\geq 2$, while maintaining the optimality of $p$. This finishes the proof of the lemma.
\end{proof}
\section{Omitted Proofs in Section~\ref{sec:totally-ordered}}
\label{sec:totally-ordered-appendix}

\subsection{Omitted proofs in Section~\ref{sec:np-hardness}}
\label{sec:np-hardness-proof}

\begin{numberedtheorem}{\ref{thm:totally-ordered-nphard}}
$\orderedip$ is strongly NP-hard, even when each realized buyer has at most three distinct item values.
\end{numberedtheorem}

\begin{proof}
We prove the theorem via a reduction from \textsc{Max-Cut}. 

For any \textsc{Max-Cut} instance with graph $G(V,E)$, let $n=|V|>180$ be large enough. Consider an instance of $\orderedip$ with $n+1$ items. For convenience, we assume that each node in $V$ also has an index in $[n]$. We want the following properties of the optimal item pricing for the instance:
\begin{enumerate}
	\item The optimal item pricing has integral item prices for each item;
	\item $p_{n+1}=6n$, and there is a set of buyer types purchasing item $n+1$ with realization probability $q_1=0.9$ that do not depend on the structure of the graph and contribute $R_1(n)$ to the total revenue;
	\item $p_{i}=3i-1$ or $3i-2$, and there is a set of buyer types purchasing items in $[n]$ with realization probability $q_2<\frac{1}{12n}q_1$ that do not depend on the structure of the graph and contribute $R_2(n)$ to the total revenue;
	\item In addition to all previous buyers, for each $(i,j)\in E$ with $i<j$, there exists a set $T_{ij}$ of buyer types, and real number $R_{ij}(n)>0$ that is irrelevant to the graph structure such that: if $p_j-p_i=3(j-i)$, then the revenue contribution from $T_{ij}$ is $R_{ij}(n)$; if $p_j-p_i\neq 3(j-i)$, then the revenue contribution from $T_{ij}$ is $R_{ij}(n)+\frac{1}{n^{10}}$. The realization probability of any buyer type is polynomially bounded by $n$ (at least $poly(n^{-1})$).
\end{enumerate}

Before going to the construction, we first show the strongly NP-hardness of $\orderedip$ for an instance with above properties. This proves the claim of the Theorem.

Given an instance with such properties, we can calculate the revenue of the optimal item pricing for the instance. The total revenue contributed from buyer types from Property 2 and 3 is $R_1(n)+R_2(n)$. For any cut $C=(V_1,V\setminus V_1)$, if each item $i$ in $V_1$ is priced $3i-1$, while each item $j$ in $V\setminus V_1$ is priced $3j-2$, the total revenue contributed from buyer types from Property 4 is $\sum_{(i,j)\in E}R_{ij}(n)+\frac{1}{n^{10}}|C|$. Thus, for a graph $G(V,E)$ with maximum cut $c_{\max}$, the corresponding instance of $\orderedip$ has maximum revenue 
\begin{equation*}
h(G)=R_1(n)+R_2(n)+\sum_{(i,j)\in E}R_{ij}(n)+\frac{1}{n^{10}}c_{\max}.
\end{equation*} 
This builds a bijection between the maximum cut of $G$, and the optimal item pricing revenue of the $\orderedip$ instance constructed from $G$. Since all inputs for the $\orderedip$ instance are polynomially bounded, we know that $\orderedip$ is strongly NP-hard from the APX-hardness of \textsc{Max-Cut}.

Now let's go back to show how to construct the $\orderedip$ satisfying all properties.

\paragraph{Property 1.} To make sure that the optimal item pricing has all integral prices, we only need to construct the distribution such that each buyer type has integral value for every item. For such an instance, if the optimal item pricing does not have integral price for every item, we can round up the price for each item to the closet integer without reducing the revenue. The reason is that such a round up procedure does not change the incentive of any buyer type. If a buyer $v$ prefers to purchase item $i$ to $j$ under pricing $p$, which means $v_i-p_i\geq v_j-p_j$, then $v_i-\lceil p_i\rceil \geq v_j-\lceil p_j \rceil$ since $v_i$ and $v_j$ are both integer. Thus we can only focus on the class of item pricing with integral item prices.

\paragraph{Property 2.} Construct a buyer type with value $v_{n+1}=6n$ for item $n+1$, and $v_{i}=0$ for all $i\leq n$. In other words, the buyer only wants to purchase item $n+1$ with value $6n$, and is not interested in any other item. The buyer type appears with probability $q_1$, where $q_1$ is very close to $1$. We will make sure that the rest of the buyer types will contribute less than $q_1$ revenue, so the optimal pricing will not set a price less than $6n$ for item $n+1$. This will be done by letting the rest of the buyer types have maximum item value $\leq 6n$ for each item and appear with probability $<\frac{1}{6n}q_1=\frac{3}{20}n$ in total. The revenue contribution of the buyer is $R_1(n)=6nq_1$ by setting $p_{n+1}=6n$.

\paragraph{Property 3.} For each $i\in[n]$, construct a buyer type $v$ with value $v_1=v_2=\cdots=v_{i-1}=0$, $v_i=v_{i+1}=\cdots=v_{n+1}=3i-1$, which appears with probability $q'_2=\frac{1}{36n^2}q_1=\frac{1}{40n^2}$; a buyer type $v'$ with value $v'_1=\cdots=v'_{i-1}=0$, $v'_i=\cdots=v'_{n+1}=3i-2$, which appears with probability $\frac{1}{3i-2}q'_2$; and a buyer $v''$ with value $v''_1=\cdots=v''_{i-1}=0$, $v''_i=\cdots=v''_{n}=3i-2$, $v''_{n+1}=6n$, which appears with probability $q'_2$. Under any item pricing, all buyers would purchase item $i$, item $n+1$, or nothing. Note that the price for item $n+1$ has been fixed to $6n$ by Property 2. We have the following cases: 
\begin{itemize}
\item If $p_{i}\leq 3i-3$, then all of the three buyer types prefer to purchase item $i$, which lead to total revenue $(3i-3)(q'_2+\frac{1}{3i-2}q'_2+q'_2)=\frac{(3i-3)(6i-3)}{3i-2}q'_2<6nq'_2$ for the three buyer types. 
\item If $p_i\geq 3i$, then buyer $v$ and buyer $v'$ cannot afford to purchase any item, while buyer $v''$ prefers to purchase item $n+1$, which lead to revenue $6nq'_2$. 
\item If $p_i=3i-1$, then buyer $v$ purchases item $i$; buyer $v'$ purchases nothing; buyer $v''$ purchases item $n+1$, which leads to total revenue $(3i-1)q'_2+6nq'_2$.
\item If $p_i=3i-2$, then the first buyer purchase item $i$, the second buyer purchases item $i$, the third buyer purchases item $n+1$, which leads to revenue $(3i-2)(q'_2+\frac{1}{3i-2}q'_2)+6nq'_2=(3i-1)q'_2+6nq'_2$.
\end{itemize}
Thus setting $p_i=3i-1$ or $3i-2$ gives the same optimal revenue for the three buyers, while setting any other price leads to a revenue loss of at least $(3i-1)q'_2$. We set $q'_2$ to be large enough such that the rest of the buyers (to be defined in Property 4) cannot contribute $q'_2$ revenue, which can be done by letting the rest of the buyer types have maximum item value $\leq 6n$ and appear with probability $<\frac{1}{6n}q'_2=\frac{1}{240}n^3$ in total. So the optimal pricing only sets $p_i=3i-1$ or $3i-2$ for item $i$.

The total revenue contribution of the buyers is $R_2(n)=\sum_{i=1}^{n}((3i-1)q'_2+6nq'_2)$. The total realization probability of the buyers added in this property is at most $3q'_2$ for each $i$, thus at most $3nq'_2=\frac{3}{40n}$. The set of the buyers and the revenue contribution only depend on $n$, and does not depend on the graph structure.

\paragraph{Property 4.} For each edge $(i,j)\in E$ with $i<j$, let $x=3i-2$, and $y=3j-2$. Consider set $T_{ij}$ of buyer types formed by the following 4 types of buyers $v^{(1)},v^{(2)},v^{(3)},v^{(4)}$:
\begin{itemize}
	\item $v^{(1)}_{1}=\cdots=v^{(1)}_{i-1}=0$, $v^{(1)}_{i}=\cdots=v^{(1)}_{j-1}=x$, $v^{(1)}_{j}=\cdots=v^{(1)}_{n+1}=y$;
	\item $v^{(2)}_{1}=\cdots=v^{(2)}_{i-1}=0$, $v^{(2)}_{i}=\cdots=v^{(2)}_{j-1}=x+1$, $v^{(2)}_{j}=\cdots=v^{(2)}_{n+1}=y+1$;
	\item $v^{(3)}_{1}=\cdots=v^{(3)}_{i-1}=0$, $v^{(3)}_{i}=\cdots=v^{(3)}_{j-1}=x$, $v^{(3)}_{j}=\cdots=v^{(3)}_{n+1}=y+1$;
	\item $v^{(4)}_{1}=\cdots=v^{(4)}_{i-1}=0$, $v^{(4)}_{i}=\cdots=v^{(4)}_{j-1}=x+1$, $v^{(4)}_{j}=\cdots=v^{(4)}_{n+1}=y$.
\end{itemize}
In other words, the four type of buyers would purchase item $i$, item $j$ or nothing, and has slightly different values for item $i$ and $j$. Our goal here is to determine the appearance probability of each buyer type in the value distribution, such that if $p_i\equiv p_j\textrm{ mod }3$, then these buyer types contribute $R_{ij}(n)$ to the total revenue; if $p_i\not\equiv p_j\textrm{ mod }3$, then these buyer types contribute $R_{ij}(n)+n^{-10}$ to the total revenue.

Let $A$ be the following $4\times 4$ outcome matrix, such that each element of the matrix corresponds to the payment of a (row) buyer under a specific (item) item pricing:

\[
  \rowind{$(p_i,p_j)$ in the pricing\ \ \ }
  \begin{array}{@{}c@{}}
    \rowind{$(x,y)$} \\ \rowind{$(x+1,y+1)$} \\ \rowind{$(x,y+1)$} \\ \rowind{$(x+1,y)$} 
  \end{array}
  \mathop{\left[
  \begin{array}{ *{5}{c} }
     \colind{y}{$(x,y)$}  &  \colind{y}{$(x+1,y+1)$}  &  \colind{y}{$(x,y+1)$}  & \colind{x}{$(x+1,y)$}  \\
     0 &  y+1  &  y+1  & x+1  \\
     x  & x &  y+1  & x \\
     y  & y & y & y  
  \end{array}
  \right]}^{
  \begin{array}{@{}c@{}}
    \rowind{$(v_i,v_j)$ of the buyer} \\ \mathstrut
  \end{array}
  }
\]
For any vector $\mathbf{z}\in \mathbb{R}_{\geq 0}^4$, $A\mathbf{z}$ correspond to the vector of the revenues of the 4 pricings $(p_i,p_j)=(x,y),(x+1,y+1),(x,y+1),(x+1,y)$, given that $z_\ell$ buyers of type $\ell$ appear. To satisfy Property 4, we need to find a vector $\mathbf{z}$ such that $A\mathbf{z}=(R_{ij},R_{ij},R_{ij}+n^{-10},R_{ij}+n^{-10})^T$ for some $R_{ij}\geq 0$.

By solving $\mathbf{a}$ for $A\mathbf{a}=(1,1,1,1)^T$ and $\mathbf{b}$ for $A\mathbf{b}=(0,0,1,1)^T$, we get 
\begin{equation*}
\mathbf{a}=(\frac{1}{y^2+y},\frac{x}{y(y+1)(y+1-x)},\frac{y-x}{y(y+1-x)},0)^T,
\end{equation*} 
%https://www.wolframalpha.com/input/?i=%28y%29*%281%2F%28y%5E2%2By%29%29%2B%28y%29*%28x%2F%28y*%28y%2B1%29*%28y%2B1-x%29%29%29%2B%28y%29*%28%28y-x%29%2F%28y*%28y%2B1-x%29%29%29%2B%28y%29*%280%29
\begin{equation*}
\mathbf{b}=(\frac{1}{y^2+y},-\frac{x^2+y(y+1)^2-xy(y+2)}{y(y+1)(y+1-x)(y-x)},\frac{y-x}{y(y+1-x)},\frac{1}{y-x})^T.
\end{equation*}
%https://www.wolframalpha.com/input/?i=%28y%29*%281%2F%28y%5E2%2By%29%29%2B%28y%29*%28%28xy%28y%2B2%29-x%5E2-y%28y%2B1%29%5E2%29%2F%28y*%28y%2B1%29*%28y%2B1-x%29*%28y-x%29%29%29%2B%28y%29*%28%28y-x%29%2F%28y*%28y%2B1-x%29%29%29%2B%28y%29*%281%2F%28y-x%29%29
Taking 
\begin{equation*}
\mathbf{z}=2n^{-10}y^3\mathbf{a}+n^{-10}\mathbf{b}=n^{-10}\left(\frac{2y^3+1}{y^2+y},\frac{2x(y-x)y^3+xy^2+2x-y^3-2y^2-y}{y(y+1)(y+1-x)(y-x)},\frac{(2y^3+1)(y-x)}{y(y+1-x)},\frac{1}{y-x}\right),
\end{equation*} 
we have $\mathbf{z}\geq \mathbf{0}$, and 
\begin{equation*}
A\mathbf{z}=(2n^{-10}y^3,2n^{-10}y^3,2n^{-10}y^3+n^{-10},2n^{-10}y^3+n^{-10}).
\end{equation*}
Thus if with probability $z_\ell$ the buyer has type $v^{(\ell)}$, the four buyer types contribute $R_{ij}(n)=2n^{-10}y^3$ revenue if $p_i\equiv p_j\textrm{ mod }3$, and $R_{ij}(n)+n^{-10}=2n^{-10}y^3+n^{-10}$ otherwise.

Now we compute the total realization probability of the four buyer types in the distribution: $\|\mathbf{z}\|_1<10y^3n^{-10}<270n^{-7}$ by the definition of $\mathbf{z}$ and $y<3n$. Since there are less than $\frac{1}{2}n^2$ edges, the total realization probability of all buyer types added in this property is less than $\frac{1}{2}n^2\cdot270n^{-7}=135n^{-5}<\frac{1}{240n^3}$ for $n>180$ as required in Property 3. The total realization property of all buyer types added in Property 3 and 4 is less than $\frac{3}{40n}+\frac{1}{240n^3}<\frac{3}{20n}$ as required in Property 2. This completes the proof of the correctness of the construction.

\end{proof}

\subsection{Omitted proofs in Section~\ref{sec:totally-ordered-ptas}}
\label{sec:proof-totally-ordered-ptas}

\begin{numberedtheorem}{\ref{thm:totally-ordered-algorithm}}
For a unit-demand buyer with totally ordered item values, there exists an algorithm running in $poly(m,n^{poly(1/\eps)},\log\range)$ time that computes an item pricing that is $(1+\eps)$-approximation in revenue to the optimal item pricing.
\end{numberedtheorem}

\begin{proof}
%[Proof of Theorem~\ref{thm:totally-ordered-algorithm}]
We prove the theorem in several steps.
\begin{enumerate}
\item There exists a near optimal item pricing where all prices are non-decreasing powers of $(1+\eps^2)$. Let $\Pi=\{(1+\eps^2)^r | r\in\Z\}\cup\{0\}$. Then for all item pricings $p$, there exists $q^{(1)}\in\Pi^n$, such that for all value functions $v$,
  \begin{align*}
    \rev_{q^{(1)}}(v) \ge (1-O(\eps))\rev_p(v). 
  \end{align*}
\item At a small loss in revenue, we can restrict prices to lie in a small set. In particular, for all value distributions $\dist$ with value range $\range$, there exists an efficiently computable set $\Pi^*\subset\Pi$ with $|\Pi^*|=poly(1/\eps,\log \range)$ such that for all item pricings $q^{(1)}\in \Pi^n$, there exists an item pricing $q^{(2)}\in {\Pi^*}^n$ satisfying
  \begin{align*}
    \rev_{q^{(2)}}(\dist) \ge (1-O(\eps))\rev_{q^{(1)}}(\dist). 
  \end{align*}
\item Next, we define for each unit demand buyer an additive value function that closely mimics it. For a unit-demand value function $v$, define $\vadd_i=v_i-v_{i-1}$ and let $\vadd$ be the value function that assigns to a set $S\subseteq [n]$ the value $\sum_{i\in S} \vadd_i$.

  We also define a new kind of pricing that we will call an {\em interval prefix pricing}. Given a partition of the $n$ items into $t$ intervals, $I_{i_0,i_1}$, $I_{i_1,i_2}$, $\cdots$, $I_{i_{t-1},i_{t}}$ with $i_0=0$ and $i_t=n$, an interval prefix pricing $q$ is a menu with $n$ options; The $j$th option allocates the set $I_{i_\ell,j}=\{i_\ell+1, i_\ell+2, \cdots, j\}$ at price $q_j$ where $i_\ell< j\le i_{\ell+1}$.

  We furthermore say that an interval prefix pricing $q$ satisfies {\em price gap} $(\gamma,\delta)$ if (1) menu options corresponding to different intervals are priced multiplicatively apart: for all $i$, $j$, and $\ell$ with $i\le i_\ell<j$, $q_j\ge (1+\eps^2)^\gamma q_i$, (2) and, menu options corresponding to any single interval are priced multiplicatively close to each other: for all $i$, $j$, and $\ell$ with $i_\ell<i<j\le i_{\ell+1}$, $q_j\le (1+\eps^2)^\delta q_i$.
  
  We show that for value distribution $\dist$ and its corresponding additive value distribution $\dist^{\oplus}$ and item pricing $q^{(2)}\in {\Pi^*}^n$, there exists an efficiently computable set $\Pi'$ with $|\Pi'|=|\Pi^*|$ and an interval prefix pricing $q^{(3)}$ with $q^{(3)}_i\in \Pi'$ for all $i\in [n]$ and price gap $(\frac{1}{\eps^2}\ln\frac{1}{\eps^2},\frac{1}{\eps^3}\ln\frac{1}{\eps^2})$, such that
  \begin{align*}
    \rev_{q^{(3)}}(\dist^{\oplus}) \ge (1-O(\eps))\rev_{q^{(2)}}(\dist).
  \end{align*}

  The converse also holds: for every unit demand value function $v$ and interval prefix pricing $q$ with price gap $(\frac{1}{\eps^2}\ln\frac{1}{\eps^2},\frac{1}{\eps^3}\ln\frac{1}{\eps^2})$, we can efficiently compute an item pricing $q^{(4)}$ such that
  \begin{align*}
    \rev_{q^{(4)}}(v) \ge (1-O(\eps)) \rev_{q}(\vadd).
  \end{align*}

\item Finally, we show that for any distribution over additive values $\vadd$ and any set $\Pi'$ of values, an optimal interval prefix pricing $q$, with $q_i\in \Pi'$ for all $i\in [n]$ and price gap $(\frac{1}{\eps^2}\ln\frac{1}{\eps^2},\frac{1}{\eps^3}\ln\frac{1}{\eps^2})$, can be found in time polynomial in $|\Pi'|$, $n^{poly(1/\eps)}$, and $m$.
  
\end{enumerate}

The algorithm can be described as follows. By the last step, we can efficiently compute the optimal interval prefix pricing $q$ with price gap $(\frac{1}{\eps^2}\ln\frac{1}{\eps^2},\frac{1}{\eps^3}\ln\frac{1}{\eps^2})$ for the distribution $\dist^{\oplus}$ over additive buyer $\vadd$ that corresponds to the unit-demand distribution $\dist$, such that all item prices are in $\Pi'$. By Step 3, we can efficiently compute an item pricing $q^{(4)}$ with $\rev_{q^{(4)}}(v) \ge (1-O(\eps)) \rev_{q}(\vadd)$. Also by the first three steps, 
\begin{eqnarray*}
\rev_{q^{(4)}}(\dist)&\geq&(1-O(\eps))\rev_{q}(\dist^{\oplus})\geq(1-O(\eps))\rev_{q^{(3)}}(\dist^{\oplus})\geq (1-O(\eps))\rev_{q^{(2)}}(\dist)\\
&\geq& (1-O(\eps))\rev_{q^{(1)}}(\dist)\geq (1-O(\eps))\rev_{p}(\dist)
\end{eqnarray*}
for the optimal item pricing $p$. $q^{(4)}$ can be found in $poly(m,n^{poly(1/\eps)},|\Pi'|)=poly(m,\log\range,n^{poly(1/\eps)})$ time. 

Now we elaborate on each step in more detail. For simplicity, we assume $\frac{1}{\eps}$ is an integer.

\paragraph{Step 1.} We first introduce a useful lemma for approximating the revenue of a mechanism $p$ via another pricing function $q$ that approximates $p$ closely \textit{pointwise}. Slightly different forms of the lemma appears in multiple papers in the simple-versus-optimal mechanism design and the menu-size complexity literature, and is attributed to Nisan\footnote{For a more detailed discussion of the history of the idea of the lemma, we would point the readers to footnote 24 of \cite{gonczarowski2021sample}.}. For completeness, we provide a proof of the lemma after the proof of Theorem~\ref{thm:totally-ordered-algorithm}.

\begin{lemma}\label{lem:Nisan}
For any $\eps>0$, let $p$ and $q$ be two pricing functions satisfying $q(\lambda)\leq p(\lambda)\leq (1+\eps) q(\lambda)$ for all random allocations $\lambda\in\Delta(2^{[n]})$. Then for scaling factor $\alpha=(1+\eps)^{-1/\sqrt{\eps}}$ and any valuation function $v$, 
\begin{equation*}
\rev_{\alpha q}(v)\geq(1-3\sqrt{\eps})\rev_{p}(v).
\end{equation*}
\end{lemma}

Now come back to Step 1 of the proof. Consider the optimal item pricing $p$. Let $q$ be the item pricing that rounds the price of each item down to the closest integral power of $(1+\eps^2)$. Then for any lottery $\lambda$, $q(\lambda)\leq p(\lambda)<(1+\eps^2)q(\lambda)$. By Lemma~\ref{lem:Nisan}, $q^{(1)}=(1+\eps^2)^{-1/\eps}q$ is an item pricing with power-of-$(1+\eps^2)$ prices for each item, and achieves $(1-O(\eps))$ fraction of the revenue of $p$. This means that focusing on finding the optimal item pricing with discretized item prices only loses a $(1-O(\eps))$-factor in revenue. 

Any item pricing is equivalent to an item pricing with non-decreasing item prices in the totally ordered setting, since for any pricing $p$, if there exist two items $i<j$ with $p_i>p_j$, then no buyer ever prefers $i$ to $j$ since item $i$ is worse in quality with higher prices. Thus setting $p_i=p_j$ keeps the incentive of all buyer types unchanged, while after finite such operations, we can get an item pricing with non-decreasing item prices and the same revenue.

\paragraph{Step 2.} It suffices to show that there exists an item pricing with all item prices being either 0 or bounded in range $[\Omega(\eps^2),\range]$, that achieves an $(1-O(\eps))$ fraction of the revenue of $q^{(1)}$. 

The proof has the same nudging idea as Lemma~\ref{lem:Nisan}. Let $q^{(2)}$ be defined as follows. For each item $i$ with $q^{(1)}_i\leq \eps^2$, $q^{(2)}_i=0$; otherwise, $q^{(2)}_i=(1-\eps)q^{(1)}_i$. Then for each buyer type $v$, assume that she purchases item $i$ under $q^{(1)}$ and $j$ under $q^{(2)}$. Since the buyer prefers $i$ over $j$ under $q^{(1)}$, we have 
\begin{equation}\label{eqn:step2ineq1}
v_i-q^{(1)}_i\geq v_{j}-q^{(1)}_j;
\end{equation}
Since the buyer prefers $j$ over $i$ under $q^{(2)}$, and $(1-\eps)(q^{(1)}-\eps^2)\leq q^{(2)}\leq (1-\eps)q^{(1)}$, we have
\begin{equation}\label{eqn:step2ineq2}
v_j-(1-\eps)(q^{(1)}_j-\eps^2)\geq v_j-q^{(2)}_j\geq v_{i}-q^{(2)}_i\geq v_i-(1-\eps)q^{(1)}_i.
\end{equation}
Add inequalities \eqref{eqn:step2ineq1} and \eqref{eqn:step2ineq2}, we have $q^{(1)}_j\geq q^{(1)}_i-\eps+\eps^2$. Thus 
\begin{equation*}
q^{(2)}_j\geq (1-\eps)(q^{(1)}_j-\eps^2)\geq(1-\eps)q^{(1)}_i-\eps+\eps^2.
\end{equation*}
Taking the expectation over $v\sim \dist$. Observe that $\rev_{q^{(1)}}(\dist)\geq 1$ by selling only item $n$ at price 1, since in the input model we assume that each valuation function has at least value 1 for some item. We have 
\begin{equation*}
\rev_{q^{(2)}}(\dist)\geq (1-\eps)\rev_{q^{(1)}}(\dist)-\eps+\eps^2=(1-O(\eps))\rev_{q^{(1)}}(\dist).
\end{equation*}
The item prices of $q^{(2)}$ are in some set $\Pi^*$ with $|\Pi^*|=O(\frac{1}{\eps}\log\frac{\range}{\eps})$ since all item prices in $q^{(2)}$ are in range $[(1-\eps)\eps^2,\range]$ and are powers of $(1+\eps^2)$ in $q^{(1)}$ multiplied by $(1-\eps)$.

\paragraph{Step 3.} For item pricing $q^{(2)}$, consider the following $\frac{1}{\eps}$ sets of prices: for each $\ell\in[\frac{1}{\eps}]$, let
\begin{equation*}
V_{\ell}=\left\{(1+\eps^2)^{s}\Big|s\in \Z,\ s\equiv a \ mod \frac{1}{\eps^3}\ln\frac{1}{\eps^2}\textrm{ for }\frac{\ell}{\eps^2}\ln\frac{1}{\eps^2}\leq a<\frac{\ell+1}{\eps^2}\ln\frac{1}{\eps^2}\right\}.
\end{equation*}
Since $V_{\ell}$ are disjoint sets with the union being the set of all power-of-$(1+\eps^2)$ prices, there exists a set of prices $V_{\ell}$ such that in item pricing $q^{(2)}$, the revenue contribution from items with price in $V_{\ell}$ is at most $\eps$ fraction of the total revenue. Consider the following item pricing $q'$: for any item with price $q^{(2)}_i\not\in V_{\ell}$, $q'_i=q^{(2)}_i$; for any item with $q^{(2)}_i\in V_{\ell}$, $q'_i$ is set to the smallest $q^{(2)}_j\not\in V_{\ell}$ with $j>i$. In other words, the prices of all items in $V_{\ell}$ are raised while the incentives of all buyer types that do not purchase an item with price in $V_{\ell}$ in $q^{(2)}$ stay unchanged. Thus the revenue of $q'$ is at least $(1-O(\eps))$ fraction of the revenue of $q^{(2)}$, $\rev_{q'}(\dist)\geq(1-O(\eps))\rev_{q^{(2)}}(\dist)$.

Now we construct an interval prefix pricing $q$ as follows. The price of each item $q_i$ is set to $q'_i$, with all items naturally grouped to intervals as follows: for any integer $r$, all items with $\log_{1+\eps^2}q_i$ in range $[\frac{r}{\eps^3}\ln\frac{1}{\eps^2}+\frac{\ell+1}{\eps^2}\ln\frac{1}{\eps^2},\frac{r+1}{\eps^3}\ln\frac{1}{\eps^2}+\frac{\ell}{\eps^2}\ln\frac{1}{\eps^2})$ are grouped to an interval. Then for any two items $i,j$ in the same interval, the logarithm (with base $1+\eps^2$) of the two prices differs by at most $\frac{1}{\eps^3}\ln\frac{1}{\eps^2}$; for any two items $i<j$ in different intervals, the prices differ by a factor at least $(1+\eps^2)^{\frac{1}{\eps^2}\ln\frac{1}{\eps^2}}>\frac{1}{\eps^2}$. Thus $q$ has price gap $(\frac{1}{\eps^2}\ln\frac{1}{\eps^2},\frac{1}{\eps^3}\ln\frac{1}{\eps^2})$.

Let $I_{i_0,i_1}$, $I_{i_1,i_2}$, $\cdots$, $I_{i_{t-1},i_{t}}$ be the interval partitioning of the entire set of items.
Now we study the price of any set $T=\{j_1,j_2,\cdots,j_b\}$ of items under such an interval prefix pricing $q$. Suppose that the item with the largest index $j_b$ is in interval $I_{i_{d-1},i_{d}}$. Then the price $q(T)$ of the set is lower-bounded by $q'_{j_b}$, and upper-bounded by $q'_{j_b}$ plus the prices of all items in intervals prior to $I_{i_{d-1},i_{d}}$. The payment for interval $I_{i_{d-2},i_{d-1}}$ is at most 
\begin{equation*}
q'_{i_{d-1}}<\eps^2 q'_{i_{d-1}+1}\leq \eps^2 q'_{j_b}.
\end{equation*}
Similarly the payment for intervals with smaller indexes are geometrically smaller, and sum up to at most $\eps^2 q'_{j_b}$. This means that $q'_{j_b}\leq q(T)\leq (1+2\eps^2)q'_{j_b}<(1+4\eps^2)q'_{j_b}$. 

On the other hand, we introduce a new mechanism $q''$ called \textit{prefix pricing}. In such a mechanism, the buyer can pay $q'_i$ to obtain a prefix set of items $\{1,2,\cdots,i\}$. For any unit-demand buyer $v$ with totally-ordered item values, she buys set $\{1,2,\cdots,i\}$ under $q''$ if and only if she buys $i$ under $q'$. Also under $q''$, buyer $v$ and $\vadd$ purchase the same set. Thus $\rev_{q''}(\vadd)=\rev_{q''}(v)=\rev_{q'}(v)$. In prefix pricing $q''$, the price of set $T$ is $q''(T)=q'_{j_b}$, which is the price of the smallest prefix that contains $j_b$. Thus 
\begin{equation}\label{eqn:totally-ordered-approx-gap}
(1+2\eps^2)^{-1}q(T)\leq q''(T)\leq q(T)\leq (1+2\eps^2)q''(T).
\end{equation} 
By Lemma~\ref{lem:Nisan}, there exists a scaling $q^{(3)}=(1+4\eps^2)^{-1/(2\eps)}q$ of interval prefix pricing $q$ which gives $(1-O(\eps))$-fraction of the revenue of $q''$ for any buyer $\vadd$. Thus
\begin{equation*}
\rev_{q^{(3)}}(\dist^{\oplus})\geq (1-O(\eps))\rev_{q''}(\dist^{\oplus})=(1-O(\eps))\rev_{q'}(\dist)\geq(1-O(\eps))\rev_{q^{(2)}}(\dist).
\end{equation*}
Since all item prices in $q$ are in $\Pi^*$, we have that all item prices in $q^{(3)}$ are in $\Pi'=\{(1+4\eps^2)^{-1/(2\eps)}y|y\in\Pi^*\}$.

The converse also holds: given any interval prefix pricing $q$ with price gap $(\frac{1}{\eps^2}\ln\frac{1}{\eps^2},\frac{1}{\eps^3}\ln\frac{1}{\eps^2})$, we can define prefix pricing $q''$ such that $q''(\{1,2,\cdots,i\})=q_i$. Then \eqref{eqn:totally-ordered-approx-gap} still holds, which means that by Lemma~\ref{lem:Nisan}, there exists a scaling $q^o=(1+4\eps^2)^{-1/(2\eps)}q''$ of prefix pricing $q''$ which gives $(1-O(\eps))$-fraction of the revenue of $q$ for any additive buyer $\vadd$. Let $q^{(4)}$ be the item pricing with $q^{(4)}_i=q^o(\{1,2,\cdots,i\})$. Then
\begin{equation*}
\rev_{q^{(4)}}(v)=\rev_{q^o}(v)=\rev_{q^o}(\vadd)\geq(1-O(\eps))\rev_{q}(v).
\end{equation*}

\paragraph{Step 4.} We aim to solve the following problem: find the optimal interval prefix pricing $q$ with price gap $(\frac{1}{\eps^2}\ln\frac{1}{\eps^2},\frac{1}{\eps^3}\ln\frac{1}{\eps^2})$ for an additive buyer, such that the item prices $q_i$ that define $q$ are in set $\Pi'$.

This can be solved via the following dynamic program. Since the buyer is additive, the revenue contribution from each interval can be calculated separately without worrying about the incentive of the buyer. Let $F[i,z]$ be the optimal revenue of the interval prefix pricing (with price gap $(\gamma,\delta)=(\frac{1}{\eps^2}\ln\frac{1}{\eps^2},\frac{1}{\eps^3}\ln\frac{1}{\eps^2})$ and without prefix buying constraints) from only items in prefix set $\{1,2,\cdots,i\}$, with the last item price being $q_i=z$. Then we can write the following recursive formula:
\begin{equation*}
F[i,z]=\max_{j<i,y\geq z(1+\eps^2)^{-\delta},w\leq y(1+\eps^2)^{-\gamma},y,z\in \Pi'}\Big\{F[j,w]+G[j,i,y,z)\Big\},
\end{equation*}
where $G[j,i,y,z)$ denotes the optimal revenue of a pricing $q$ which sells a prefix of interval $I_{j,i}=\{j+1,j+2,\cdots,i\}$ with price $q_{j+1}=y$ and $q_i=q(\{j+1,j+2,\cdots,i\})=z$ to an additive buyer with distribution $\dist^{\oplus}$. In the recursive formula, we want all items in $I_{j,i}$ to be priced in range $[y,z]$ with $y\geq z(1+\eps^2)^{-\delta}$, with all items in $\{1,2,\cdots,j\}$ to be priced at most $w\leq y(1+\eps^2)^{-\gamma}$. The objective we want to solve is $\max_{z\in \Pi'}F[n,z]$. 

Now we analyze the running time of the dynamic program. Each $i,j,y,z$ in the recursive formula has $poly(m,n,\frac{1}{\eps})$ possibilities. The running time of calculating each $G[j,i,y,z)$ is $poly(m,n^\delta)=poly(m,n^{\frac{1}{\eps^3}\ln\frac{1}{\eps^2}})$, since there are only at most $\delta$ different prices for the items in $I_{j,i}$, thus at most $poly(n^\delta)$ different non-decreasing pricings. Therefore the dynamic program can be solved in $poly(m,n^{poly(1/\eps)},|\Pi'|)$ time.

\end{proof}

\begin{numberedlemma}{\ref{lem:Nisan}}
For any $\eps>0$, let $p$ and $q$ be two pricing functions satisfying $q(\lambda)\leq p(\lambda)\leq (1+\eps) q(\lambda)$ for all random allocations $\lambda\in\Delta(2^{[n]})$. Then for scaling factor $\alpha=(1+\eps)^{-1/\sqrt{\eps}}$ and any valuation function $v$, 
\begin{equation*}
\rev_{\alpha q}(v)\geq(1-3\sqrt{\eps})\rev_{p}(v).
\end{equation*}
\end{numberedlemma}

\begin{proof}[Proof of Lemma~\ref{lem:Nisan}]
Suppose that under pricing $p$ the buyer $v$ purchases lottery $\lambda$, while under pricing $p'=\alpha q$ the buyer purchases $\lambda'$. Then since the buyer has higher utility purchasing $\lambda$ than $\lambda'$ under $p$,
\begin{equation*}
v(\lambda)-p(\lambda)\geq v(\lambda')-p(\lambda').
\end{equation*}
Since the buyer has higher utility purchasing $\lambda'$ than $\lambda$ under $p'$,
\begin{equation*}
v(\lambda')-p'(\lambda')\geq v(\lambda)-p'(\lambda).
\end{equation*}
Adding the above two inequalities, we have
\begin{equation*}
p(\lambda')-p'(\lambda')\geq p(\lambda)-p'(\lambda).
\end{equation*}
By $p'(\lambda')=\alpha q(\lambda')\geq \alpha(1+\eps)^{-1}p(\lambda')$, and $p'(\lambda)=\alpha q(\lambda)\leq \alpha p(\lambda)$, from the above inequality we have
\begin{equation}\label{eqn:Nisanapprox}
(1-\alpha(1+\eps)^{-1})p(\lambda')\geq p(\lambda')-p'(\lambda')\geq p(\lambda)-p'(\lambda)\geq (1-\alpha)p(\lambda).
\end{equation}
Notice that 
\begin{eqnarray*}
\alpha=(1+\eps)^{-1/\sqrt{\eps}}<(1+\eps)e^{-\sqrt{\eps}}<(1+\eps)(1-\sqrt{\eps}+\eps)<1-\sqrt{\eps}+2\eps,
\end{eqnarray*}
and 
\begin{eqnarray*}
\alpha=(1+\eps)^{-1/\sqrt{\eps}}>e^{-\sqrt{\eps}}>1-\sqrt{\eps}.
\end{eqnarray*}
Apply the bound of $\alpha$ to \eqref{eqn:Nisanapprox}, we have
\begin{equation*}
p(\lambda')\geq\frac{1-\alpha}{1-\alpha(1+\eps)^{-1}}p(\lambda)>\frac{1-(1-\sqrt{\eps}+2\eps)}{1-(1-\sqrt{\eps})(1+\eps)^{-1}}p(\lambda)>1-3\sqrt{\eps}p(\lambda).
\end{equation*} 
Thus $\rev_{\alpha q}(v)\geq(1-3\sqrt{\eps})\rev_{p}(v)$.

\end{proof}

\subsection{Omitted proofs in Section~\ref{sec:totally-ordered-additive}}
\label{sec:proof-totally-ordered-additive}

\begin{numberedtheorem}{\ref{thm:totally-ordered-approx-additive}}
  For any additive buyer with totally ordered value for all items, item pricing cannot achieve an approximation ratio better than $O(\log\log n)$ to the optimal (deterministic) buy-many mechanism.
\end{numberedtheorem}

\begin{proof}
  
  For an additive buyer, an ordered item setting means that every buyer with value $v$, $v_1\leq\cdots\leq v_n$. Let $\subadditiverev(\dist)$ be the revenue from the optimal deterministic buy-many mechanism (i.e. monotone subadditive pricing) for a distribution $\dist$ over additive buyers. Let $\srev(\dist)$ be the revenue from the optimal item pricing for distribution $\dist$.
  
  We show this theorem by contradiction. More precisely, we show that if for any additive buyer with ordered item values, item pricing can approximate the optimal deterministic buy-many revenue with factor $o(\log\log n)$, then for any additive buyer item pricing can approximate the optimal deterministic buy-many revenue with factor $o(\log n)$, which contradicts Theorem 4.4 of \cite{chawla2019buy}. A restatement of the theorem is as follows:
  
  \begin{theorem}\label{thm:ec19-thm44}[Theorem 4.4 of \cite{chawla2019buy}]
  There exists a distribution $\dist_0$ over additive buyer types with $n$ items such that no item pricing algorithm can get a revenue more than $1/o(\log n)$ fraction of the revenue from the optimal deterministic buy-many mechanism. Moreover, the item values of each are either 0 or in range $[2n^{-1/6}, 1]$, while $\subadditiverev(\dist_0)\geq n^{1/6}$.
  \end{theorem}
  
  Now consider this distribution from the lemma. From $\dist_0$, we construct distribution $\dist_1$ as follows: for every buyer type $v$ in the support of $\dist_0$, there is a corresponding buyer type $v'= (v_1+\frac{1}{n},\ldots, v_n+\frac{1}{n})$ in the support of $\dist_1$ with the same realization probability. In other words, the buyer's value for each item increases by $\frac{1}{n}$ in $\dist_1$. Then, we construct $\dist_2$ being a distribution over additive valuation functions over $N=M^{n-1}+M^{n-2}+\cdots+1=\frac{M^n-1}{M-1}=O(n^{2n-2})$ items using $\dist_1$ as follows: Let $M=n^2$, and $i_j=M^{n-1}+M^{n-2}+\cdots+M^{n-j}$ for $1\leq j\leq n$.
  For every buyer type $v'$ in the support of $\dist_1$, we ``split'' each item $i$ to $M^{n-i}$ identical items with values scaled down by a factor of $M^{i-n}$ as follows: item 1 is split to $i_1=M^{n-1}$ items $1,2,\cdots,M^{n-1}$ with value $v''_1=v''_2=\cdots=v''_{i_1}=\frac{v'_1}{M^{n-1}}$; item 2 is split to $M^{n-2}$ items $i_1+1,i_1+2,\cdots,i_2$ with value $v''_{i_1+1}=v''_{i_1+2}=\cdots=v''_{i_2}=\frac{v'_2}{M^{n-2}}$; and so on. So for any buyer type $v'$ from the support of distribution $\dist_1$, we get a buyer $v''$ in the support of $\dist_2$ with the following valuation function over $N$ items with the same realization probability:
  
  $$v''= \left(\underbrace{\frac{v'_1}{M^{n-1}}, \cdots, \frac{v'_1}{M^{n-1}}}_{M^{n-1}\textrm{ many}}, \underbrace{\frac{v'_2}{M^{n-2}}, \cdots, \frac{v'_2}{M^{n-2}}}_{M^{n-2}\textrm{ many}}, \cdots, \underbrace{\frac{v'_{n-1}}{M},\cdots, \frac{v'_{n-1}}{M}}_{M \text{ many}}, v'_n\right).$$
  
  Note that $\dist_2$ is a distribution over additive buyer types with ordered item values. This is because for any buyer type $v''$ in the support of $\dist_2$, the corresponding $v$ in the support of $\dist_1$ satisfies $\frac{1}{n}\leq v'_i\leq 1+\frac{1}{n}$. Then for any two items $i$ and $i+1$ in $[N]$, either $v''_i=v''_{i+1}$; or $v''_{i}=\frac{v'_j}{M^{n-j}}$ and $v''_{i+1}=\frac{v'_{j+1}}{M^{n-j-1}}$ for some $j$ in $[n]$, then $v''_i<v''_{i+1}$ as $\frac{v'_{j}}{v'_{j+1}}<n^2=M$.
  
  Now assume that there exists an item pricing for $\dist_2$ that approximates the optimal deterministic buy-many revenue with a factor of $o(\log \log N)$. Let $p$ be the optimal item pricing for $\dist_2$. 
  This means $$\frac{\subadditiverev(\dist_2)}{\rev_p(\dist_2)}=o(\log \log N)=o(\log n).$$
  
  Now we show the following two properties:
  \begin{enumerate}
    \item $\srev(\dist_2)\leq \srev(\dist_0)+1$: Note that without loss of generality we can assume that items $i_{j-1}+1,i_{j-1}+2,\cdots,i_{j}$ that are ``split'' from item $j$ have the same price in $p$, as they are identical items and can have the same optimal item prices for an additive buyer. Also we can assume that the item price $p_i$ for each item in $[N]$ is at least $\frac{1}{n}$, since each buyer type has value at least $\frac{1}{n}$ for each item. Consider the following item pricing $q$ for $\dist_0$: for each item $j\in[n]$, $q_j=M^{j-1}p_{i_j}-\frac{1}{n}=\sum_{i_{j-1}+1\leq i\leq i_j}p_i-\frac{1}{n}$. In other words, for the set of items split from item $j$ in $\dist_2$, their prices in $p$ are aggregated to $q_j$, while the additional $\frac{1}{n}$ added to the item values is removed. Then for any buyer type $v''$ in the support of $\dist_2$ (and its corresponding type $v$ in the support of $\dist_0$), she can afford to purchase item $i\in[i_{j-1}+1,i_j]$ (thus all items in $[i_{j-1}+1,i_j]$) if and only if $v$ can afford to purchase item $j\in[n]$. Furthermore, if $v''$ can afford to purchase items in $[i_{j-1}+1,i_j]$, her payment $M^{j-1}p_{i_j}$ for these items is exactly $\frac{1}{n}$ larger than the payment of the corresponding $v$ for item $j\in[n]$. Thus $\rev_p(\dist_2)\leq \rev_{q}(\dist_0)+n\cdot\frac{1}{n}\leq \srev(\dist_0)+1$.
    \item $\subadditiverev(\dist_2)\geq \subadditiverev(\dist_0)$: For any subadditive pricing $q$ over $\dist_0$, we construct the following subadditive pricing $p$ over $\dist_2$. For any set $S\subseteq [n]$, let $S'=\bigcup_{j\in S}\{i_{j-1}+1,i_{j-1}+2,\cdots,i_j\}\subseteq [N]$ be the set of items split from items in $S$, then $p(S'):=q(S)+\frac{|S|}{n}$. Then $p$ is a subadditive pricing, and the incentives of $v$ and $v''$ are ``identical'' under the two subadditive pricings, as $v(S)-q(S)=v''(S')-p(S')$. Therefore, if $v$ chooses to purchase $S$ under $q$, $v''$ will purchase $S'$ under $p$. As $p(S')\geq q(S)$, we have $\rev_{p}(\dist_2)\geq \rev_{q}(\dist_0)$. By letting $q$ be the optimal subadditive pricing for $\dist_0$, we have $\subadditiverev(\dist_2)\geq \subadditiverev(\dist_0)$.
  \end{enumerate}
  
  Using the above properties we have:

  $$\frac{\subadditiverev(\dist_0)}{\srev(\dist_0)}\leq\frac{\subadditiverev(\dist_2)}{\srev(\dist_2)-1}=o(\log n)$$
  since $\srev(\dist_2)\geq\frac{1}{o(\log n)}\subadditiverev(\dist_2)$, and $\subadditiverev(\dist_2)\geq\subadditiverev(\dist_0)\geq n^{1/6}$. This contradicts Theorem~\ref{thm:ec19-thm44}, which concludes the proof.
  
\end{proof}

\section{Omitted Proofs in Section~\ref{sec:partially-ordered}}
\label{sec:partially-ordered-appendix}

\begin{numberedtheorem}{\ref{thm:general-value-approx-algorithm}}
For a general-valued buyer with partially ordered values, if the partially ordered set containing all items has width $k$, then there exists an algorithm running in $poly(m,n^{poly(k,1/\eps)},\log\range)$ time that computes an item pricing that is $(1+\eps)$-approximation in revenue to the optimal item pricing.
\end{numberedtheorem}

\begin{proof}

The algorithm is similar to the totally-ordered setting, but we need to define the generalized prefixes and intervals.
\begin{itemize}
	\item \textit{Prefix}: For any set $T\in\feasiblesets$, the prefix parameterized by $T$ is a set of items $P_T\supseteq T$ such that item $j\in[n]$ is in $P_T$ if and only if there exists $i\in T$ such that $j\preceq i$. Such a definition ensures that for $T_1,T_2\in\feasiblesets$ with $T_1\preceq T_2$, $P_{T_1}\subseteq P_{T_2}$.

	For example, in the $k$-category setting, a set $T\in\feasiblesets$ has the $i_j$-th item from each category $j$. The prefix $P_T$ defined by $T$ includes the first $i_j$ items in each category $j$. It can be viewed as a more general $k$-dimensional prefix compared to the single-category setting.

	\item \textit{Interval}: For two sets $T,T'\in\feasiblesets$ with $T\subseteq T'$, interval $I_{T,T'}=P_{T'}\setminus P_{T}$ is a set of items with contiguous item types between $T$ and $T'$ in the ordering graph. 

	For example, in the $k$-category setting, an interval contains the $(i'_j+1)$-th to the $i_j$-th items with $i'_j\leq i_j$ in each category $j$. It can be viewed as a more general $k$-dimensional interval compared to the single-category setting.
\end{itemize}
\noindent Now we are ready to describe the proof sketch in steps.
\begin{enumerate}
\item There exists a near optimal item pricing where all prices are powers of $(1+\eps^2)$. Let $\Pi=\{(1+\eps^2)^r | r\in\Z\}\cup\{0\}$. Then for all item pricings $p$, there exists $q^{(1)}\in\Pi^n$, such that for all value functions $v$,
  \begin{align*}
    \rev_{q^{(1)}}(v) \ge (1-O(\eps))\rev_p(v). 
  \end{align*}
Furthermore, without loss of generality, we can assume that for the item pricing $q^{(1)}$ we consider, set $\{i|q^{(1)}_i\leq y\}$ is a prefix for any $y\in\R$.
\item At a small loss in revenue, we can restrict prices to lie in a small set. In particular, for all value distributions $\dist$ with value range $\range$, there exists an efficiently computable set $\Pi^*\subset\Pi$ with $|\Pi^*|=poly(1/\eps,k,\log n, \log \range)$ such that for all item pricings $q^{(1)}\in \Pi^n$, there exists an item pricing $q^{(2)}\in {\Pi^*}^n$ satisfying
  \begin{align*}
    \rev_{q^{(2)}}(\dist) \ge (1-O(\eps))\rev_{q^{(1)}}(\dist). 
  \end{align*}
\item Given a partition of the $n$ items into $t$ intervals, $I=(I_{T_0,T_1},I_{T_1,T_2},\cdots,I_{T_{t-1},T_{t}})$ with $T_0=\emptyset$, and $T_t$ be a set of items that dominates all other items (with $P_{T_t}=[n]$). We furthermore say that for any interval partition $I$, an item pricing $q$ satisfies {\em price gap} $(\gamma,\delta)$ if (1) items corresponding to different intervals are priced multiplicatively apart: for all $i$, $j$, and $\ell$ with $i\in P_{T_{\ell}}$ and $j\not\in P_{T_{\ell}}$, $q_j\ge (1+\eps^2)^\gamma q_i$, (2) and, menu options corresponding to any single interval are priced multiplicatively close to each other: for all $i$, $j$, and $\ell$ with $i,j\in I_{\ell}$, $q_j\le (1+\eps^2)^\delta q_i$.

We show that for value distribution $\dist$ and item pricing $q^{(2)}\in {\Pi^*}^n$, there exists an item pricing $q^{(3)}$ with $q^{(3)}_i\in \Pi^*$ for all $i\in [n]$ and price gap $(\gamma,\delta)=(\frac{1}{\eps^2}\ln\frac{k}{\eps^2},\frac{k}{\eps^3}\ln\frac{k}{\eps^2})$, such that
  \begin{align*}
    \rev_{q^{(3)}}(\dist) \ge (1-O(\eps))\rev_{q^{(2)}}(\dist).
  \end{align*}

\item We define a new kind of pricing that we will call an {\em interval prefix pricing}. Given a partition of the $n$ items into $t$ intervals $I=(I_{T_0,T_1},I_{T_1,T_2},\cdots,I_{T_{t-1},T_{t}})$, an interval prefix pricing $q$ is a mechanism defined by a vector of item prices $(q_1,q_2,\cdots,q_n)$: For any set $S_{\ell}\subseteq I_{\ell}$ of items, there is a menu allocating a set of items $P^*_{S_{\ell}}=\cup_{1\leq j\leq \ell-1}T_{j}\cup S_{\ell}$, with price $q(P^*_{S_{\ell}})=\sum_{i\in P^*_{S_{\ell}}}q_i$. In other words, to purchase any set of items in $S_{\ell}$, the buyer also needs to purchase all sets of items $T_{1},\cdots,T_{\ell-1}$ that define the previous intervals.

We show that for every value function $v$ and item pricing $q^{(3)}\in {\Pi^*}^n$ with price gap $(\gamma,\delta)=(\frac{1}{\eps^2}\ln\frac{k}{\eps^2},\frac{k}{\eps^3}\ln\frac{k}{\eps^2})$, there exists an efficiently computable set $\Pi'$ with $|\Pi'|=|\Pi^*|$ and an interval prefix pricing $q^{(4)}$ with $q^{(4)}_i\in \Pi'$ for all $i\in [n]$ and price gap $(\frac{1}{\eps^2}\ln\frac{k}{\eps^2},\frac{k}{\eps^3}\ln\frac{k}{\eps^2})$, such that
  \begin{align*}
    \rev_{q^{(4)}}(v) \ge (1-O(\eps))\rev_{q^{(3)}}(v).
  \end{align*}

The converse is also true: for every value function $v$ and interval prefix pricing $q$ with price gap $(\frac{1}{\eps^2}\ln\frac{k}{\eps^2},\frac{k}{\eps^3}\ln\frac{k}{\eps^2})$, we can efficiently compute an item pricing $q^{(5)}$ such that
  \begin{align*}
    \rev_{q^{(5)}}(v) \ge (1-O(\eps)) \rev_{q}(v).
  \end{align*}

\item We define for each arbitrary-valued buyer $v$ and interval partitioning $I$ an \textit{additive-over-intervals} value function that closely mimics it. For an arbitrary value function $v$, for any set $S=S_1\cup S_2\cup \cdots\cup S_t$ of items with $S_\ell\subseteq I_\ell$ for every $\ell\in[t]$, define
\begin{equation*}
	\vaddi(S)=\sum_{\ell=1}^{t}\big(v(S_\ell\cup T_{\ell-1})-v(T_{\ell-1})\big).
\end{equation*}
In other words, $\vaddi(S_{\ell})$ is the value gain of getting set $S_{\ell}$, when the buyer $v$ has set of items $T_{\ell-1}$ at hand. We write $\dist_I^{\oplus}$ as the distribution of $\vaddi$ corresponding to $v\sim \dist$. $v$ and $\vaddi$ has the same behavior under an interval prefix pricing defined by interval partition $I$. When the interval partition is clear from the context, we will omit $I$ and write $\vadd=\vaddi$.

We show that for every additive-over-intervals value function $\vadd$ and interval prefix pricing $q^{(4)}$ with $q^{(4)}_i\in \Pi'$, and with price gap $(\gamma,\delta)=(\frac{1}{\eps^2}\ln\frac{k}{\eps^2},\frac{k}{\eps^3}\ln\frac{k}{\eps^2})$, there exists an efficiently computable set $\Pi^o$ with $|\Pi^o|=|\Pi^*|$ and an item $q^{(6)}$ with $q^{(6)}_i\in \Pi^o$ for all $i\in [n]$ and price gap $(\frac{1}{\eps^2}\ln\frac{k}{\eps^2},\frac{k}{\eps^3}\ln\frac{k}{\eps^2})$, such that
  \begin{align*}
    \rev_{q^{(6)}}(\vadd) \ge (1-O(\eps))\rev_{q^{(4)}}(v).
  \end{align*}

The converse also holds: for every value function $v$ and item pricing $q$ with price gap $(\frac{1}{\eps^2}\ln\frac{k}{\eps^2},\frac{k}{\eps^3}\ln\frac{k}{\eps^2})$, we can efficiently compute an interval prefix pricing $q^{(7)}$ with price gap $(\frac{1}{\eps^2}\ln\frac{k}{\eps^2},\frac{k}{\eps^3}\ln\frac{k}{\eps^2})$ such that
  \begin{align*}
    \rev_{q^{(7)}}(v) \ge (1-O(\eps)) \rev_{q}(\vadd).
  \end{align*}

\item Finally, we show that for any distribution $\dist$ over arbitrary values and any set $\Pi^o$ of values, an optimal item pricing $q$ for value distribution $\dist_I^{\oplus}$ and the corresponding interval partition $I$, with $q_i\in \Pi^o$ for all $i\in [n]$ and price gap $(\frac{1}{\eps^2}\ln\frac{k}{\eps^2},\frac{k}{\eps^3}\ln\frac{k}{\eps^2})$, can be found in time polynomial in $|\Pi^o|$, $n^{k\cdot poly(1/\eps)}$, and $m$.
  
\end{enumerate}

The algorithm can be described as follows. By the last step, we can efficiently compute the optimal item pricing $q$ with price gap $(\frac{1}{\eps^2}\ln\frac{k}{\eps^2},\frac{k}{\eps^3}\ln\frac{k}{\eps^2})$ and the optimal interval partitioning for the distribution $\dist^{\oplus}_I$ over additive buyer $\vaddi$ that corresponds to the general value distribution $\dist$, such that all item prices are in $\Pi^o$. By Step 4 and 5, we can efficiently compute an item pricing $q^{(5)}$ with $\rev_{q^{(5)}}(v) \ge (1-O(\eps)) \rev_{q}(\vadd)$. Also by the first five steps, 
\begin{eqnarray*}
\rev_{q^{(5)}}(\dist)&\geq&(1-O(\eps))\rev_{q}(\dist^{\oplus})\geq(1-O(\eps))\rev_{q^{(6)}}(\dist^{\oplus})\geq (1-O(\eps))\rev_{q^{(4)}}(\dist)\\
&\geq& (1-O(\eps))\rev_{q^{(3)}}(\dist)\geq (1-O(\eps))\rev_{q^{(2)}}(\dist)\\
&\geq&(1-O(\eps))\rev_{q^{(1)}}(\dist)\geq (1-O(\eps))\rev_{p}(\dist)
\end{eqnarray*}
for the optimal item pricing $p$. $q^{(5)}$ can be found in $poly(m,n^{k poly(1/\eps)},|\Pi'|)=poly(m,\log\range,n^{k poly(1/\eps)})$ time.  

Now we elaborate on each step in more detail.

\paragraph{Step 1.} This is the same as Step 1 in the totally-ordered setting, as any item pricing is $(1+\eps^2)$-approximated pointwise by a power-of-$(1+\eps^2)$ item pricing.

Without loss of generality, we can assume that for the item pricing $q^{(1)}$ we consider, set $\{i|q^{(1)}_i\leq y\}$ is a prefix for any $y\in\R$. This is because we can assume for any two items $i\prec j$, $q^{(1)}_i\leq q^{(1)}_j$.

\paragraph{Step 2.} This is similar to Step 2 in the totally-ordered setting. It suffices to show that there exists an item pricing with all item prices being either 0 or bounded in range $[\Omega(\frac{1}{kn^k}\eps^2),\range]$, that achieves an $(1-O(\eps))$ fraction of the revenue of $q^{(1)}$. 

Firstly observe that $\rev_{q^{(1)}}(\dist)\geq\frac{1}{n^k}$. This is because by definition of the input distribution, for any $v\sim \dist$, there exists a set $T\in\feasiblesets$ with $v(T)\geq 1$. Then by $|\feasiblesets|<n^k$, there exists a set $T\in\feasiblesets$ with $\Pr_{v\sim\dist}[v(T)\geq 1]\geq\frac{1}{n^k}$. Price each item in $T$ at $\frac{1}{|T|}$, and all other items at price $+\infty$, the mechanism sells $T$ at price 1 with probability $\geq\frac{1}{n^k}$, thus has revenue at least $\frac{1}{n^k}$.

Let $q^{(2)}$ be defined as follows. For each item $i$ with $q^{(1)}_i\leq \frac{1}{kn^k}\eps^2$, $q^{(2)}_i=0$; otherwise, $q^{(2)}_i=(1-\eps)q^{(1)}_i$. Then for each buyer type $v$, assume that she purchases set $T$ under $q^{(1)}$ and $T'$ under $q^{(2)}$. Since the buyer prefers $T$ over $T'$ under $q^{(1)}$, we have 
\begin{equation}\label{eqn:generalstep2ineq1}
v(T)-q^{(1)}(T)\geq v(T')-q^{(1)}(T);
\end{equation}
Since the buyer prefers $T'$ over $T$ under $q^{(2)}$, and $(1-\eps)(q^{(1)}(S)-k\cdot\frac{1}{kn^k}\eps^2)\leq q^{(2)}(S)\leq (1-\eps)q^{(1)}(S)$ for any set $S\in\feasiblesets$, we have
\begin{equation}\label{eqn:generalstep2ineq2}
v(T')-(1-\eps)(q^{(1)}(T')-\frac{1}{n^k}\eps^2)\geq v(T')-q^{(2)}(T')\geq v(T)-q^{(2)}(T)\geq v(T)-(1-\eps)q^{(1)}(T).
\end{equation}
Add inequalities \eqref{eqn:generalstep2ineq1} and \eqref{eqn:generalstep2ineq2}, we have $q^{(1)}(T')\geq q^{(1)}(T)-\frac{1}{n^k}\eps+\frac{1}{n^k}\eps^2$. Thus 
\begin{equation*}
q^{(2)}(T')\geq (1-\eps)(q^{(1)}(T)-\frac{1}{n^k}\eps^2)>(1-\eps)q^{(1)}(T)-\frac{1}{n^k}\eps.
\end{equation*}
Taking the expectation over $v\sim \dist$, and by $\rev_{q^{(1)}}(\dist)\geq \frac{1}{n^k}$ we have 
\begin{equation*}
\rev_{q^{(2)}}(\dist)\geq (1-\eps)\rev_{q^{(1)}}(\dist)-\frac{1}{n^k}\eps=(1-O(\eps))\rev_{q^{(1)}}(\dist).
\end{equation*}
The item prices of $q^{(2)}$ are in some set $\Pi^*$ with $|\Pi^*|=O(\log_{1+\eps^2} (kn^k\range/\eps^2))=poly(k,\log n,\log\range,1/\eps)$ since all item prices in $q^{(2)}$ are in range $[(1-\eps)\frac{1}{kn^k}\eps^2,\range]$ and are powers of $(1+\eps^2)$ in $q^{(1)}$ multiplied by $(1-\eps)$.

\paragraph{Step 3.} This is almost the same as the first part of the proof of Step 3 in the totally ordered case. Consider the following $\frac{k}{\eps}$ sets of prices: for each $\ell\in[\frac{k}{\eps}]$, let
\begin{equation*}
V_{\ell}=\left\{(1+\eps^2)^{s}\Big|s\in \Z,\ s\equiv\ a \ mod\ \frac{k}{\eps^3}\ln\frac{k}{\eps^2}\textrm{ for }\frac{\ell}{\eps^2}\ln\frac{k}{\eps^2}\leq a<\frac{\ell+1}{\eps^2}\ln\frac{k}{\eps^2}\right\}.
\end{equation*}
Then there exists $\ell$ such that in item pricing $q^{(2)}$, the revenue contribution from buyer types that purchase a set $T\in\feasiblesets$ (with size k) that includes an item $i\in T$ with $q^{(2)}_i\in V_{\ell}$ is at most $\eps$ fraction of the total revenue of $p$. 
% Increase the price of each item $i$ with $q^{(2)}_i\in V_{\ell}$ until some value $q^{(3)}_i$ not in $V_{\ell}$, while keep $q^{(3)}_i=q^{(2)}_i$ for all other items, 
Let $q^{(3)}$ be the following item pricing: for each item $i$, $q^{(3)}_i$ is the smallest power of $(1+\eps^2)$ that is at least $q^{(2)}_i$ and not in $V_{\ell}$.
We have an item pricing $q^{(3)}$ with all prices not in $V_{\ell}$, while achieves $(1-\eps)$ fraction of the revenue of $q^{(2)}$.

Now we construct the interval partition $I$ as follows. All items are naturally grouped to intervals as follows: for any integer $r$, all items with $\log_{1+\eps^2}q^{(3)}_i$ in range $[\frac{kr}{\eps^3}\ln\frac{k}{\eps^2}+\frac{\ell+1}{\eps^2}\ln\frac{k}{\eps^2},\frac{k(r+1)}{\eps^3}\ln\frac{k}{\eps^2}+\frac{\ell}{\eps^2}\ln\frac{k}{\eps^2})$ are grouped to an interval. Then for any two items $i,j$ in the same interval, the logarithm (with base $1+\eps^2$) of the two prices differs by at most $\frac{k}{\eps^3}\ln\frac{k}{\eps^2}$; for any two items $i<j$ in different intervals, the prices differ by a factor at least $(1+\eps^2)^{\frac{1}{\eps^2}\ln\frac{k}{\eps^2}}<\frac{k}{\eps^2}$. Thus $q^{(3)}$ has price gap $(\frac{1}{\eps^2}\ln\frac{k}{\eps^2},\frac{k}{\eps^3}\ln\frac{k}{\eps^2})$.
 
\paragraph{Step 4.} For any item pricing algorithm $q^{(3)}$ with price gap $(\gamma,\delta)=(\frac{1}{\eps^2}\ln\frac{k}{\eps^2},\frac{k}{\eps^3}\ln\frac{k}{\eps^2})$, let $q$ denote the corresponding interval prefix pricing with $q_i=q^{(3)}_i$ for each item $i$. Consider any set $S\in \feasiblesets$ with $S=S_1\cup S_2\cup\cdots\cup S_{t_1}$, with $S_\ell\in I_{\ell}$ for each $\ell$, and $S_{t_1}$ being non-empty. Then for prefix pricing $q$,
\begin{equation*}
q(S)=\sum_{i\in S_{t_1}}q^{(3)}_i+\sum_{j\in \cup_{\ell\leq t_1-1} T_{\ell}}q^{(3)}_j\in[1,1+O(\eps)]\sum_{i\in S_{t_1}}q^{(3)}_i
\end{equation*}
since for any $j\in T_{\ell-a}$ and $i\in S_{t_1}$, $q^{(3)}_j<\frac{1}{k}\eps^a q^{(3)}_i$ by the definition of $(\gamma,\delta)$ interval price gap and $|T_{\ell-a}|\leq k$. On the other hand, for item pricing $q^{(3)}$ without prefix buying constraint,
\begin{equation*}
q^{(3)}(S)=\sum_{i\in S_{t_1}}p_i+\sum_{j\in \cup_{\ell\leq t_1-1} S_{\ell}}p_j\in[1,1+O(\eps)]\sum_{i\in S_{t_1}}p_i.
\end{equation*}
Therefore 
\begin{equation}\label{eqn:partially-ordered-gap}
q(S)\in [1-O(\eps),1+O(\eps)] q^{(3)}(S).
\end{equation}  
Thus $q$ pointwise $(1+O(\eps))$-approximates $q^{(3)}$. By Lemma~\ref{lem:Nisan} there exists a scaling $q^{(4)}=(1+O(\eps)^2)^{-1/O(\eps)}q$ of $q$ that achieves $(1-O(\eps))$ fraction of the revenue of $q^{(3)}$. Each item price $q^{(4)}_i$ is uniformly scaled from $q_i=q^{(3)}_i\in \Pi^*$, thus there exists an efficiently computable set $\Pi'$ with $\Pi'=|\Pi^*|$ such that $q^{(4)}_i\in \Pi'$, and $\rev_{q^{(4)}}(v)\geq (1-O(\eps))\rev_{q^{(3)}}(v)$

The converse also holds: for any interval prefix pricing $q$ with price gap $(\frac{1}{\eps^2}\ln\frac{k}{\eps^2},\frac{k}{\eps^3}\ln\frac{k}{\eps^2})$, define $q^{(3)}$ be the item pricing with $q^{(3)}_i=q_i$. Then \eqref{eqn:partially-ordered-gap} still holds, which means that by Lemma~\ref{lem:Nisan} there exists a scaling $q^{(5)}=(1+O(\eps)^2)^{-1/O(\eps)}q^{(3)}$ of $q^{(3)}$ that achieves $(1-O(\eps))$ fraction of the revenue of $q$, and can be computed efficiently.

\paragraph{Step 5.} For interval prefix pricing $q^{(4)}$, consider item pricing $q$ with $q_i=q^{(4)}$. For any additive-over-intervals buyer $\vadd$, consider any set $S$ with $S=S_1\cup S_2\cup\cdots\cup S_{t_0}$, where $S_\ell\in I_{\ell}$, and $S_{t_0}\neq\emptyset$. Then using the same reasoning as Step 4 that the predominant part of the $q^{(4)}(S)$ comes from $S_{t_0}$,
\begin{equation*}
q^{(4)}(S)=\sum_{i\in S_{t_0}}q^{(4)}_i+\sum_{j\in \cup_{\ell\leq t_0-1} T_{\ell}}q^{(4)}_j\in[1,1+O(\eps)]\sum_{i\in S_{t_0}}q^{(4)}_i,
\end{equation*}
while
\begin{equation*}
q(S)=\sum_{i\in S_{t_0}}q^{(4)}_i+\sum_{j\in \cup_{\ell\leq t_0-1} S_{\ell}}q^{(4)}_j\in[1,1+O(\eps)]\sum_{i\in S_{t_0}}q^{(4)}_i.
\end{equation*}
Thus 
\begin{equation}\label{eqn:partially-ordered-gap2}
q(S)\in[1-O(\eps),1+O(\eps)]q^{(4)}(S)
\end{equation} 
for any set $S$ that can be purchased by $\vadd$. By Lemma~\ref{lem:Nisan} there exists a scaling $q^{(6)}=(1+O(\eps)^2)^{-1/O(\eps)}q$ of $q$ that achieves $(1-O(\eps))$ fraction of the revenue of $q^{(4)}$. 
Then
\begin{equation*}
\rev_{q^{(6)}}(\vadd)\geq (1-O(\eps))\rev_{q^{(4)}}(\vadd)=(1-O(\eps))\rev_{q^{(4)}}(v),
\end{equation*}
Here the last equality follows from $v$ and $\vadd$ having the same behavior under any interval prefix pricing. 
Since all item prices in $q^{(4)}$ are in an efficiently computable set $\Pi'$, all item prices in $q^{(6)}$ are also in some efficiently computable set $\Pi^o$ with $|\Pi^o|=|\Pi'|=|\Pi^*|$. 

The converse also holds: for any interval prefix pricing $q$ with price gap $(\frac{1}{\eps^2}\ln\frac{k}{\eps^2},\frac{k}{\eps^3}\ln\frac{k}{\eps^2})$, define $q^{(4)}$ to be the item pricing with $q^{(4)}_i=q_i$. Then \eqref{eqn:partially-ordered-gap2} still holds, which means that by Lemma~\ref{lem:Nisan} there exists a scaling $q^{(7)}=(1+O(\eps)^2)^{-1/O(\eps)}q^{(4)}$ of $q^{(4)}$ that achieves $(1-O(\eps))$ fraction of the revenue of $q$, and can be computed efficiently.

\paragraph{Step 6.} We aim to solve the following problem: find the optimal item pricing $q$ with interval price gap $(\gamma,\delta)=(\frac{1}{\eps^2}\ln\frac{k}{\eps^2},\frac{k}{\eps^3}\ln\frac{k}{\eps^2})$ and the corresponding interval partition $I$ for the additive-over-intervals buyer distribution $\dist^{\oplus}_I$, such that the item prices $q_i$ that define $q$ are in set $\Pi^o$.

Generalizing the totally-ordered case, we can solve this via the following dynamic program. Since the buyer is additive over intervals, the revenue contribution from each interval can be calculated separately without worrying about the buyer's incentive. Let $F[i,z]$ be the optimal revenue of the item pricing (with interval price gap $(\gamma,\delta)=(\frac{1}{\eps^2}\ln\frac{k}{\eps^2},\frac{k}{\eps^3}\ln\frac{k}{\eps^2})$) from items in prefix $P_{T_i}$, with the maximum item price in $P_{T_i}$ being $z$. Then we can write the following recursive formula:
\begin{equation*}
F[T,z]=\max_{T'\prec T,y\geq z(1+\eps^2)^{-\delta},w\leq y(1+\eps^2)^{-\gamma},T'\in \feasiblesets,y\in \Pi^o}\Big\{F[T,w]+G[T',T,y,z)\Big\},
\end{equation*}
where $G[T',T,y,z)$ denotes the optimal revenue of a pricing $q$ which sells interval $I_{T',T}=P_{T}\setminus P_{T'}$ with item prices between $y=\min_{i\in I_{T',T}}q_i$ and $z=\max_{i\in I_{T',T}}q_i$ to an additive-over-intervals buyer with values determined by the corresponding general-valued buyer with distribution $\dist$. The objective we want to solve is $\max_{z\in \Pi^o,T_{\max}:P_{T_{\max}}=[n]}F[T_{\max},z]$, where $T_{\max}\in\feasiblesets$ is a set of items that can be demanded by the buyer such that $P_{T_{\max}}=[n]$. 

Now we analyze the running time of the dynamic program. Each $T,T',y,z$ in the recursive formula has $poly(m,n^{k)},|\Pi^o|)$ possibilities. Now we show that the running time of calculating each $G[T',T,y,z)$ is $poly(m,n^{k\delta})=poly(m,n^{\frac{k^3}{\eps^3}\ln\frac{k}{\eps^2}})$. There are only $s\leq \delta$ different prices for the items in $I_{T',T}$. For each price $y=y_1\leq y_2\leq\cdots\leq y_{s}=z$, set $S_{y_j}=\{i:q_i\leq y_j\}$ must be a prefix of the interval, thus at most $n^k$ possibilities. Thus there are at most $poly(n^{k\delta})$ possibilities of sets $S_{y_1}\subseteq S_{y_2}\subseteq\cdots\subseteq S_{y_s}$, which corresponds to $poly(n^{k\delta})$ different non-decreasing pricings. Therefore the dynamic program can be solved in $poly(m,n^{poly(k,1/\eps)},|\Pi^o|)$ time.

\end{proof}

\begin{numberedlemma}{\ref{lem: support-size}}
Let $V$ be a set of non-negative real numbers. Then we can efficiently find a set of power-of-$(1+\eps^2)$ prices $V'$ with $|V'|=O(|V|^2\frac{1}{\eps^2}\ln\frac{1}{\eps^2})$ satisfying the following: For any buyer that is unit-demand over $n$ items such that for any buyer type $v$ and item $i$ the buyer has value $v_{i}\in V$, the optimal power-of-$(1+\eps^2)$ item pricing $p$ satisfies $p_i\in V'$ for any $i\in[n]$. 
\end{numberedlemma}

\begin{proof}[Proof of Lemma~\ref{lem: support-size}]
Define $V'$ as follows: $V'$ contains all value $z$ being power of $1+\eps^2$, such that there exists $y,y'\in V\cup\{0\}$ with $\frac{1}{1+\eps^2}z\leq y-y'\leq \frac{1+\eps^2}{\eps^2}z$. Then $|V'|=O(|V|^2\frac{1}{\eps^2}\ln\frac{1}{\eps^2})$. 

Now we show that there exists an optimal power-of-$(1+\eps^2)$ item pricing $p^*$ such that $p^*_i\in V'$ for each item $i\in[n]$. Consider any optimal power-of-$(1+\eps^2)$ item pricing $p$. Without loss of generality, assume all item prices in $p$ are at most the largest value in $V$, since no buyer type can afford a higher price. If there exists some item $i$ such that $p_i\not\in V'$, consider the set of items $S$ with the lowest such price, and a new item pricing $p'$, which multiplies the prices of the items in $S$ by $(1+\eps^2)$; in other words, $p'_i=(1+\eps^2)p_i$ for $i\in S$, and $p'_i=p_i$ otherwise. 

Consider any buyer type $v$, and assume that the buyer purchases item $i$ under $p$, and item $j$ under $p'$ ($j$ may not exist, and in this case $v_j=0$). Now we show that the buyer's payment does not decrease under $p'$ compared to $p$. If $p_i>p_j$, since the buyer switches to purchase $j$ under $p'$, $i\in S$. Since the buyer has a higher utility purchasing $i$ than $j$ under $p$, we have $v_i-p_i\geq v_j-p_j$, thus $v_i-v_j\geq p_i-p_j\geq\frac{\eps^2}{1+\eps^2}p_i$ as $p_i$ and $p_j$ are both powers of $1+\eps^2$. Since the buyer has a higher utility purchasing $j$ than $i$ under $p'$, we have $v_j-p'_j\geq v_i-p'_i=v_i-(1+\eps^2)p_i$, thus $v_i-v_j\leq (1+\eps^2)p_i$. Therefore $v_i\in V'$ by letting $y=v_i$ and $y'=v_j$ in the definition of $V'$, which contradicts the assumption that $i\in S$. Thus $p_i\leq p_j$, then $p'_j\geq p_j\geq p_i$, which means the payment of the buyer does not decrease under $p'$.

Consider the process that changes the item pricing from $p$ to $p'$. In this process, the revenue of the mechanism does not decrease. By repeating such a process we can finally get an item pricing with prices all in $V'$, while keeping the optimality.

\end{proof}

\end{document}